\RequirePackage{snapshot} % for bundledoc
\documentclass[pdflatex,sn-mathphys-num]{sn-jnl}% Math and Physical Sciences Author Year Reference Style
\usepackage{graphicx}%
\usepackage{multirow}%
\usepackage{amsmath,amssymb,amsfonts}%
\usepackage{amsthm}%
\usepackage{mathrsfs}%
\usepackage[title]{appendix}%
\usepackage{xcolor}%
\usepackage{textcomp}%
\usepackage{manyfoot}%
\usepackage{booktabs}%
\usepackage{algorithm}%
\usepackage{algorithmicx}%
\usepackage{algpseudocode}%
\usepackage{listings}%
\usepackage{graphicx}
\usepackage{float}
\usepackage{subcaption}
\usepackage{optidef}
\theoremstyle{thmstyleone}%
\newtheorem{theorem}{Theorem}%  meant for continuous numbers
\newtheorem{lemma}[theorem]{Lemma}
\newtheorem{corollary}[theorem]{Corollary}

\theoremstyle{thmstyletwo}%
\newtheorem{example}[theorem]{Example}%
\newtheorem{remark}[theorem]{Remark}%

\theoremstyle{thmstylethree}%
\newtheorem{definition}[theorem]{Definition}%
\newtheorem{assumption}[theorem]{Assumption}%

\usepackage{tikz}
\usetikzlibrary{angles,quotes}
\usetikzlibrary{decorations.markings,intersections}
\usetikzlibrary{decorations.pathreplacing}
\tikzset{->-/.style 2 args={
    postaction={decorate},
    decoration={markings, mark=at position #1 with {\arrow[ultra thick, #2]{>}}}
  },
  ->-/.default={0.5}{}
}

\usepackage{todonotes}

\newcommand{\bbF}{\mathbb{F}}
\newcommand{\bbR}{\mathbb{R}}
\newcommand{\Zpm}{Z_1^{\pm 1}}
\newcommand{\Bpm}{B_1^{\pm 1}}

\newcommand{\mvec}[1]{\vec{#1}}
\newcommand{\abs}[1]{|#1|}

\DeclareMathOperator{\mmod}{mod}
\DeclareMathOperator{\rank}{rank}
\DeclareMathOperator{\Conv}{Conv}
\newcommand{\tpmod}[1]{{(\mmod #1)}}

\DeclareMathOperator{\ima}{im}

\usepackage{fancyvrb}
\VerbatimFootnotes

\usepackage{siunitx}
\usepackage{placeins}

\begin{document}

\title[Angle-optimization and Simplification]{On Angle-optimization and Simplification of
  Degree-$1$ Homology Representatives}


\author*[1]{\fnm{Emerson G.} \sur{Escolar}}\email{e.g.escolar@people.kobe-u.ac.jp}

\author[1]{\fnm{Yuta} \sur{Shimada}}
% \email{246d502d@stu.kobe-u.ac.jp}

\affil[1]{\orgdiv{Graduate School of Human Development and Environment},
  \orgname{Kobe University},
  \orgaddress{\street{3-11 Tsurukabuto, Nada},
    \city{Kobe City},
    \postcode{657-8501},
    \state{Hyogo},
    \country{Japan}}}


\abstract{
  In topological data analysis, in particular persistent homology analysis,
  extracting ``optimal'' representatives for homology classes is crucial for identifying geometric regions of interest.
  In prior work, optimality is defined in terms of minimizing length or volume.
  In this work, we restrict our attention to a single homology class in degree $1$
  and introduce
  the total absolute curvature of cycles as the cost function.
  We show that this cost function, based on angles between edges of cycles,
  penalizes departures from planarity, convexity, and simple-ness of the cycle representative.
  We formulate the ``angle-optimal homologous cycle problem'', recast it as a binary quadratic optimization problem,
  and show the results of experiments on artificial toy data.
  }




\maketitle

\section{Introduction}
\label{sec:introduction}

In topological data analysis, the use of persistent homology
\citep{landi1997new,frosini1999size,robins1999towards,edelsbrunner2002topological}
as a descriptor for the topological features, in particular holes, of data
has been successfully applied in various fields.

In terms of homology, topological features are described using
cycles, but only up to homology (i.e.\ cycles deemed homologous are not distinguished).
However, in applications,
% of homology and by extension persistent homology,
it is important to be able to distinguish between cycles based on
various geometric quantities, in particular, by finding a representative cycle among
homologous cycle that optimizes some particular geometric quantity.
Furthermore, one can consider not just changing the representative cycle of a given homology class,
but also allowing for optimizing the set of homology generators (i.e.\ homology classes) under consideration.
In prior work, the problem has been considered for both homology cycles and persistent homology,
and optimality is defined in terms of
minimizing length (or weight) \cite{chambers2009minimum,erickson2005greedy,dey2010optimal,escolar2014computing,escolar2016ocphlp,emmett2016multiscale,wu2017optimal,dey2018persistent},
or enclosed volume \cite{obayashi2018volume,obayashi2023stable},
relative to covers \cite{zomorodian2008localized},
using enclosing balls \cite{chen2010measuring},
or others notions of size.
We note that a line of research recasts the problem(s) as linear (or integer) programming
\cite{dey2010optimal,escolar2014computing,escolar2016ocphlp,obayashi2018volume,obayashi2023stable},
and this is close to the approach taken in this work.
We also refer to the paper \cite{li2021minimal} for comparisons
of some the methods based on linear programming.
% Above, we do not claim to have given a comprehensive overview, but mention
% some references.

Here, instead of minimal length or enclosed volume, we are motivated by the
problem finding planar (or close-to-planar) representatives.
We explain the reasoning for this motivation as follows.
Recall that the Jordan curve theorem states that a planar simple closed curve divides the plane into two regions: the interior bounded by the curve and the exterior.
Suppose that degree-$1$ representative cycle in $\bbR^N$ is planar
(i.e.\ it sits in a plane in $\bbR^N$).
Then, we can interpret the interior region it bounds within its plane
as a part of the hole that it describes,
potentially leading to further applications and visualization methods.

In this initial work in this direction,
we restrict our attention to formulating the problem and testing our formulation.
We consider a single homology class in degree-$1$ simplicial homology,
and use a cost function based on angles between consecutive edges,
adapting the total absolute curvature of \cite{milnor1950total,taniyama1998total} to our setting.
We show that this cost function penalizes departures from planarity, convexity, and simple-ness of the cycle representative.
We formulate the ``angle-optimal homologous cycle problem'', recast it as a binary quadratic optimization problem,
and show the results of experiments on artificial toy data.

We also note that \cite{taniyama1998total} and subsequent works \cite{kobayashi1998curvature,nagasaka1998vertex,nagasaka2000graphs}
studies polygonal maps $G \rightarrow \mathbb{R}^N$ with minimal total absolute curvature for a given finite graph $G$,
i.e.\ minimizing over such polygonal maps.
While the minimization of the total absolute curvature is a common theme,
we emphasize that, in contrast,
in our formulation the vertices of the simplicial complex is given a fixed embedding in $\mathbb{R}^N$,
and we minimize among cycles homologous to an input cycle.

This paper is organized as follows.
In Section~\ref{sec:background} we recall some definitions and basic facts needed.
In Section~\ref{sec:formulation} we show that
the total absolute curvature can be expressed by a quadratic form (Lemma~\ref{lemma:quadraticform})
and using this, we show that
given a decomposition of a $1$-cycle $z$ in a simplicial complex into simple cycles,
its total absolute curvature $\kappa(z)$ decomposes
as a sum of the total absolute curvatures of the simple cycles
plus a term penalizing shared vertices between the simple cycles (Theorem~\ref{thm:decompo}).
In the same section we formulate the
angle-optimal homologous cycle problem (AOHCP) which aims to
minimize the total absolute curvature among cycles homologous to the input $1$-cycle.
We show that this can be expressed in a standard form of a binary quadratic optimization problem
(Problem~\eqref{opt:StandardBQP}).
In Section~\ref{sec:computations}, using a commerical solver,
we perform some computational experiments of solving AOHCP for some toy examples,
and make some observations about the results of the computational demonstrations.
In Section~\ref{sec:discussion} we summarize our findings and note some directions for future research.

%%% Local Variables:
%%% mode: LaTeX
%%% TeX-master: "main"
%%% End:

\section{Background}
\label{sec:background}

First we recall some basic terminology for simplicial homology.

Let $V$ be a finite set.
Recall that an abstract simplicial complex over $V$ is a set $K \subseteq 2^V \setminus \{\emptyset\}$
satisfying the conditions that
if $v \in V$ then $\{v\} \in K$,
and
if $\sigma \in K$ and $\emptyset \neq \tau \subseteq \sigma$ then $\tau \in K$.
Elements $\sigma$ of a simplicial complex $K$ are called \emph{simplices}, and
the elements of a simplex are called its \emph{vertices}.
For $\sigma \in K$, $\tau$ satisfying $\emptyset \neq \tau \subseteq \sigma$
is said to be a \emph{face} of $\sigma$\footnote{Thus, the condition
  ``if $\sigma \in K$ and $\emptyset \neq \tau \subseteq \sigma$ then $\tau \in K$''
  can be expressed by saying that $K$ is closed under the face relation.}.
The \emph{dimension} of a simplex $\sigma$ is defined to be $\#\sigma -1$, i.e.\ its number of vertices minus $1$.
If the dimension of $\sigma$ is $k$, then $\sigma$ is also said to be a $k$-simplex.
The set of $k$-simplices of $K$ is denoted by $K_k$.
In particular, $K_0$ is also called the \emph{vertex set} of $K$, and is in bijection with $V$
(we shall freely identify $v \in V$ with $\{v\} \in K_0$, and $V$ with $K_0$ under this bijection).

\begin{assumption}
  Throughout this work, we require that the abstract simplicial complex $K$ has vertex set $V \subset \bbR^N$.% , that is,
  % each vertex is in fact a point in some Euclidean space $\bbR^N$.
\end{assumption}

However, we need \emph{not} assume that $K$ is a \emph{geometric simplicial complex}. We recall the following.
A set $S$ in $\bbR^N$ is said to be \emph{convex} if for any $x, y \in S$, the line segment from $x$ to $y$ lies in $S$.
The \emph{convex hull} of $S \subset \bbR^N$, denoted $\Conv(S)$, is defined to be the smallest convex set containing $S$.
A \emph{geometric simplicial complex} $\Sigma$ with vertex set $V \subset \bbR^N$ is defined as follows.
First,
a (geometric) simplex is
the convex hull $\Conv(\sigma)$
of some $\emptyset \neq \sigma \subseteq V$
where the elements of $\sigma$ are affinely independent.
The elements of $\sigma$ are called the vertices of the simplex $\Conv(\sigma)$.
For $\emptyset \neq \tau \subseteq \sigma$ with $\sigma$ affinely independent,
$\Conv(\tau)$ is called a \emph{face} of $\Conv(\sigma)$.
Then, a collection $\Sigma$ of (geometric) simplices  is said to be a geometric simplicial complex if
it satisfies the following conditions: it is closed under the face relation,
and if two simplices $\Conv(\sigma)$ and $\Conv(\sigma')$ have nonempty intersection,
then the intersection is a face of both $\Conv(\sigma)$ and $\Conv(\sigma')$.
Given a geometric simplicial complex $\Sigma$,
one obtains an abstract simplicial complex
$\{\sigma \subset V| \Conv(\sigma) \in \Sigma\}$
by only retaining the information about the vertices of each geometric simplex.

We return to basic definitions relating to an abstract simplicial complex $K$.
Consider the total orderings of the vertices of a $k$-simplex $\sigma = \{v_0, v_1, \hdots, v_k\}$ of $K$.
Two total orders are said to be equivalent if they differ by an even permutation. Under this equivalence relation,
if $k > 0$ then the total orderings of the vertices of a $k$-simplex $\sigma$ are divided into two equivalence classes,
while if $k=0$ then there is only one equivalence class.
An equivalence class of orderings of the vertices of $\sigma$
is called an \emph{orientation} of the simplex $\sigma$. An \emph{oriented  $k$-simplex} is a $k$-simplex together with a choice of orientation.
An oriented simplex
(with orientation given by the equivalence class of the order $v_0, v_1, \hdots, v_k$)
will be denoted by $[v_0, v_1, \hdots, v_k]$.

Let $\bbF$ be a field.
For each integer $k$,
the \emph{$k$th chain group} with $\bbF$ coefficients of a simplicial complex $K$,
denoted $C_k(K, \bbF)$,
is
the $\bbF$-vector space freely generated by the oriented $k$-simplices of $K$ modulo the relations
$[v_0, v_1, \hdots, v_k] = -[p(v_0), p(v_1), \hdots, p(v_k)]$ where $p$ is any odd permutation
of the vertices of $\sigma = \{v_0, v_1, \hdots, v_k\}$ for $\sigma \in K_k$.
Elements of $C_k(K, \bbF)$ are called \emph{$k$-chains}.
For each $k$-simplex, arbitrarily choose an orientation.
Then, the set of (equivalence classes of) oriented $k$-simplices with the chosen orientations
forms a basis for $C_k(K, \bbF)$.

The \emph{$k$th boundary map} of $K$ is the $\bbF$-linear map
$\partial_k : C_k(K, \bbF) \rightarrow C_{k-1}(K, \bbF)$
defined by linear extension of
\[
  \partial_k([v_0, v_1, \hdots, v_k]) = \sum_{i=0}^k (-1)^i [v_0, \hdots, \hat{v_i}, \hdots, v_k]
\]
for oriented simplices $[v_0, v_1, \hdots, v_k]$, and where $\hat{v_i}$ means to exclude the vertex $v_i$.
It can be checked that $\partial_k \partial_{k+1}  = 0$ for all $k$, and so
$B_k(K, \bbF) :=\ima \partial_{k+1}$
is an $\bbF$-linear subspace of
$Z_k(K, \bbF) := \ker \partial_k$.
Elements of $\ker \partial_k$ are called \emph{$k$-cycles},
while elements of $\ima \partial_{k+1}$ are called \emph{$k$-boundaries}.
The \emph{$k$th homology group with $\bbF$ coefficients} is the quotient vector space
\[
  H_k(K, \bbF) = Z_k(K, \bbF) / B_k(K, \bbF) = \ker \partial_k / \ima \partial_{k+1}.
\]
For $z \in Z_k(K, \bbF)$, its \emph{homology class} is
$[z] := z + B_k(K,\bbF) \in H_k(K, \bbF)$.
Two $k$-cycles $z, z' \in Z_k(K,\bbF)$ are said to be \emph{homologous}, denoted $z \sim z'$,
if they have equal homology classes.
Note that for compatibility with our formulation as a
(binary) quadratic optimization problem in Section~\ref{sec:formulation},
we choose $\bbF = \bbR$ as the base field and shall be mostly concerned
with cycles with coefficients in $\{-1,0,1\}$.

\begin{definition}
  Let $K$ be a simplicial complex.
  A $1$-chain $c \in C_1(K, \bbR)$
  is said to be a \emph{simple cycle} if it can be written as
  \[
    c = \sum_{i=0}^{\ell-1} [v_i, v_{i+1 \tpmod{\ell}}]
  \]
  with $\ell \geq 3$ and where $v_0, v_1, \hdots, v_{\ell-1}$ are pairwise distinct vertices of $K$.
\end{definition}
A direct computation shows that if $c$ is a simple cycle, then $c \in \ker \partial_1$,
and thus a simple cycle is indeed a $1$-cycle.
We also refer to the simple cycle $c$
by the sequence of vertices $(v_0, v_1, \hdots, v_{\ell-1})$ or any of its circular shifts.
In what follows, we suppress the $\tpmod{\ell}$ notation,
and increment/decrement operations on indices of simple cycles should be understood modulo $\ell$.

Recall that we assume that the simplicial complex $K$ has vertex set $V \subset \bbR^N$.
For nonzero $x, y \in \bbR^N$
% the angle between $x$ and $y$ (i.e.\
the angle between the vectors $\overrightarrow{0x}$ and $\overrightarrow{0y}$ where $0$ is the origin
is given by the formula
\[
  \theta(x,y) := \arccos \left(\frac{\langle x, y \rangle}{\lVert x \rVert \lVert y \rVert}\right)
  \in [0,\pi].
\]
We also recall that points $x_0, x_1, \hdots, x_n \in \bbR^N$ are \emph{coplanar}
if and only if $\rank{[x_1-x_0 ~~ x_2-x_0 ~~ \hdots ~~ x_n-x_0]} \leq 2$.

\begin{definition}
  \label{defn:simplecycleangles}
  For a simple cycle $(v_0, v_1, \hdots, v_{\ell-1})$,
  its \emph{interior angle at vertex $v_i$} is defined to be
  \[
    \beta(c, v_i) := \theta(v_{i+1} -  v_i, v_{i-1} - v_i),
  \]
  while
  its \emph{exterior angle (or turning angle) at vertex $v_i$} is
  \[
    \alpha(c, v_i) := \pi - \beta(c, v_i) =  \theta(v_{i+1} -  v_i, v_i - v_{i-1}).
  \]
  Since the vertices of a simple cycle are pairwise distinct,
  the above angles are well-defined.
\end{definition}

See Figure~\ref{ex:SimpleCycleAngles} for an illustration.
\begin{figure}[h]
  \centering
  \begin{tikzpicture}[scale=1.2,baseline=(B.center)]
    \coordinate (X) at (-1.9,0.4);
    \coordinate[label=below:$v_{i-1}$] (A) at (-1.5,0);
    \coordinate[label=below:$v_{i}$]  (B) at (0,0);
    \coordinate[label=left:$v_{i+1}$]  (C) at (1,1.5);
    \coordinate (Y) at (1.4,1.3);
    \coordinate (D) at (2,0);

    \foreach \x in {A,B,C}
    \fill (\x) circle (2pt);

    \begin{scope}[every path/.style=solid,thick]
      \draw[dotted] (X) -- (A);
      \draw[->-] (A) --(B);
      \draw[->-] (B) --(C);
      \draw[dotted] (C) -- (Y);
      \draw[dashed] (B) -- (D);
      \path pic["$\beta(c{,}v_i)$", draw, angle radius=3mm, angle eccentricity=1.8, pic text options={shift={(-10pt,0pt)}}] {angle = C--B--A};
      \path pic["$\alpha(c{,}v_i)$", draw, angle radius=6mm, angle eccentricity=1.5, pic text options={shift={(+10pt,0pt)}}] {angle = D--B--C};
    \end{scope}
  \end{tikzpicture}
  \caption{Exterior angle $\alpha(c{,}v_i)$ and interior angle $\beta(c{,}v_i)$ for some simple cycle $c$ at vertex $v_u$.\\~}
  \label{ex:SimpleCycleAngles}
\end{figure}
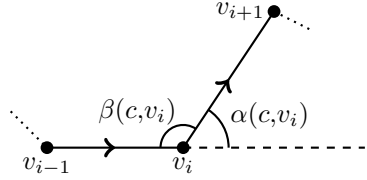

We note that the interior angle defined in Definition~\ref{defn:simplecycleangles} satisfies $\beta(c,v_i) \in [0, \pi]$
and warn that the terminology of ``interior'' and ``exterior'' in Definition~\ref{defn:simplecycleangles}
does not refer to the interior or exterior of a simple planar polygon.
See Figure~\ref{ex:InteriorNonConvex}.
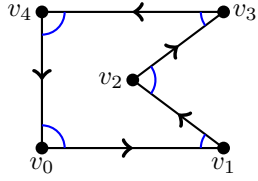
\begin{figure}[h]
  \centering
  \begin{tikzpicture}[scale=1.2,baseline=(A.center)]
    \coordinate[label=below:$v_0$]  (A) at (0,0);
    \coordinate[label=below:$v_1$] (B) at (2,0);
    \coordinate[label=left:$v_2$] (C) at (1,0.75);
    \coordinate[label=right:$v_3$] (D) at (2,1.5);
    \coordinate[label=left:$v_4$]  (E) at (0,1.5);
    \foreach \x in {A,B,C,D,E}
    \fill (\x) circle (2pt) ;

    \begin{scope}[every path/.style=solid,thick]
      \draw[->-] (A) --(B);
      \draw[->-] (B) --(C);
      \draw[->-] (C) --(D);
      \draw[->-] (D) --(E);
      \draw[->-] (E) --(A);
    \end{scope}

    \path pic["",draw,angle radius=3mm,angle eccentricity=1.5,pic text options={shift={(+2pt,+2pt)}}, blue, thick] {angle = B--A--E};
    \path pic["",draw,angle radius=3mm,angle eccentricity=1.8, blue, thick] {angle = C--B--A};
    \path pic["",draw,angle radius=3mm,angle eccentricity=1.5, blue, thick] {angle = B--C--D};
    \path pic["",draw,angle radius=3mm,angle eccentricity=1.8, blue, thick] {angle = E--D--C};
    \path pic["",draw,angle radius=3mm,angle eccentricity=1.5,pic text options={shift={(+2pt,-2pt)}}, blue, thick] {angle = A--E--D};
  \end{tikzpicture}
  \caption{The interior angles $\beta(c,v_i)$ (marked blue) for a cycle $c$ that describes a simple planar polygon
    may not correspond to the notion of internal angle in the usual sense.
    For this example, the internal angle at $v_i$ is the same as the interior angle $\beta(c,v_i)$ for $i=0,1,3,4$,
    but the internal angle at $v_2$ is equal to $2\pi - \beta(c,v_2)$.}
  \label{ex:InteriorNonConvex}
\end{figure}

\begin{definition}[{\cite{milnor1950total}}]
  \label{defn:simplecycletac}
  Let $c$ be a simple cycle.
  The \emph{total absolute curvature} of $c$, denoted $\kappa(c)$,
  is the sum of the exterior angles at its vertices.
\end{definition}
That is, for the simple cycle
  $
  c = \sum_{i=0}^{\ell-1} [v_i, v_{i+1}],
  $
  % and recalling that operations on indices of simple cycles should be understood modulo its length $\ell$:
  \[
    % \kappa(c) = \sum_{i=0}^{\ell-1} \theta(v_{i+1} -  v_i, v_i - v_{i-1})
    \kappa(c)
    = \sum_{i=0}^{\ell-1} \alpha(c,v_i)
    = \ell\pi - \sum_{i=0}^{\ell-1} \beta(c,v_i).
  \]

\begin{example}
    Below, we give an example of a simple cycle $c$ with
\[
  \kappa(c)  = 2\pi
             = \frac{\pi}{2} + \frac{\pi}{4} + \frac{\pi}{2} + \frac{\pi}{4} + \frac{\pi}{2}
             = 5\pi - \left(\frac{\pi}{2} + \frac{3\pi}{4} + \frac{\pi}{2} + \frac{3\pi}{4} + \frac{\pi}{2} \right)
\]
\[
  \begin{tikzpicture}[scale=1.2,baseline=(C.center)]
    \coordinate[label=left:$v_0$]  (A) at (0,0);
    \coordinate (Ae) at (0,-0.5);
    \coordinate[label=below:$v_1$] (B) at (2,0);
    \coordinate (Be) at (2.5,0);
    \coordinate[label=below:$v_2$] (C) at (3,1);
    \coordinate (Ce) at (3.5,1.5);
    \coordinate[label=right:$v_3$] (D) at (2,2);
    \coordinate (De) at (1.5,2.5);
    \coordinate[label=above:$v_4$]  (E) at (0,2);
    \coordinate (Ee) at (-0.5,2);

    \foreach \x in {A,B,C,D,E}
    \fill (\x) circle (2pt) ;

    \begin{scope}[every path/.style=solid, thick]
      \draw[->-] (A) --(B);
      \draw[->-] (B) --(C);
      \draw[->-] (C) --(D);
      \draw[->-] (D) --(E);
      \draw[->-] (E) --(A);
      \draw[dashed] (B) -- (Be);
      \draw[dashed] (C) -- (Ce);
      \draw[dashed] (D) -- (De);
      \draw[dashed] (E) -- (Ee);
      \draw[dashed] (A) -- (Ae);
    \end{scope}

    \path pic["$\frac{\pi}{4}$", draw, thick,angle radius=6mm, angle eccentricity=1.5, pic text options={shift={(+5pt,0pt)}}] {angle = Be--B--C};
    \path pic["$\frac{\pi}{2}$", draw, thick,angle radius=4mm, angle eccentricity=1.5, pic text options={shift={(0pt,+3pt)}}] {angle = Ce--C--D};
    \path pic["$\frac{\pi}{4}$", draw, thick,angle radius=6mm, angle eccentricity=1.5, pic text options={shift={(-2pt,0pt)}}] {angle = De--D--E};
    \path pic["$\frac{\pi}{2}$", draw, thick,angle radius=4mm, angle eccentricity=1.5, pic text options={shift={(-2pt,-2pt)}}] {angle = Ee--E--A};
    \path pic["$\frac{\pi}{2}$", draw, thick,angle radius=4mm, angle eccentricity=1.5, pic text options={shift={(+2pt,-2pt)}}] {angle = Ae--A--B};
  \end{tikzpicture}
  \quad\quad
  \begin{tikzpicture}[scale=1.2,baseline=(C.center)]
    \coordinate[label=left:$v_0$]  (A) at (0,0);
    \coordinate[label=below:$v_1$] (B) at (2,0);
    \coordinate[label=right:$v_2$] (C) at (3,1);
    \coordinate[label=right:$v_3$] (D) at (2,2);
    \coordinate[label=above:$v_4$]  (E) at (0,2);
    \foreach \x in {A,B,C,D,E}
    \fill (\x) circle (2pt) ;

    \begin{scope}[every path/.style=solid,thick]
      \draw[->-] (A) --(B);
      \draw[->-] (B) --(C);
      \draw[->-] (C) --(D);
      \draw[->-] (D) --(E);
      \draw[->-] (E) --(A);
    \end{scope}

    \path pic["$\frac{\pi}{2}$",draw,angle radius=3mm,angle eccentricity=1.5,pic text options={shift={(+2pt,+2pt)}}] {angle = B--A--E};
    \path pic["$\frac{3\pi}{4}$",draw,angle radius=3mm,angle eccentricity=1.8] {angle = C--B--A};
    \path pic["$\frac{\pi}{2}$",draw,angle radius=3mm,angle eccentricity=1.5] {angle = D--C--B};
    \path pic["$\frac{3\pi}{4}$",draw,angle radius=3mm,angle eccentricity=1.8] {angle = E--D--C};
    \path pic["$\frac{\pi}{2}$",draw,angle radius=3mm,angle eccentricity=1.5,pic text options={shift={(+2pt,-2pt)}}] {angle = A--E--D};
  \end{tikzpicture}
\]
\end{example}

In order to state Fenchel's Theorem (Theorem~\ref{thm:kappaineq}),
we recall some additional definitions following Milnor~\cite{milnor1950total}.
First, we note that a simple cycle $c$ given by $(v_\ell=v_0,v_1,\hdots,v_{\ell-1})$
describes a closed curve $\gamma_c$; i.e.\ a continuous map
$\gamma_c: [0,1] \rightarrow \mathbb{R}^N$
with some $0 = t_0 < t_1 < \hdots < t_{\ell-1} < t_\ell = 1$ such that
$v_i = \gamma_c(t_i)$ for each $i=0,1,\hdots,\ell$,
and $\gamma_c([t_i, t_{i+1}])$ traces out the line segment from $v_i$ to $v_{i+1}$,
for each $i=0,1,\hdots,\ell-1$.
Since that $\gamma_c(0) = v_0 = \gamma_c(1)$,
we actually have a continuous map from the quotient space
$S^1 = [0,1]/\{0\sim 1\}$ which is the circle.
We ignore the distinction between different parametrizations of $\gamma_c$ and circular shifts of the vertices of $c$.

\begin{remark}
  \label{remark:simpleclosedcurve}
  While it is not a notion that we need in the rest of this paper, we warn of the following potential for confusion.
  Recall that in general, a closed curve $\gamma$ is said to be \emph{simple}
  if $\gamma(t) = \gamma(s)$ only when $t=s$ or $\{t,s\} \subset \{0,1\}$.
  We warn that a simple cycle $c$ in an abstract simplicial complex
  does not necessarily give rise to closed curve $\gamma_c$ that is simple,
  depending on how its vertices are located in $\mathbb{R}^N$.
  For example, on the plane,
  let $c = [v_0,v_1] + [v_1,v_2] + [v_2,v_3] + [v_3,v_0]$ be a cycle with
  the locations of its vertices as illustrated below.
  \begin{center}
    \begin{tikzpicture}[scale=1.2,baseline=(B.center)]
      \coordinate[label=left:$v_{0}$] (A) at (0,0);
      \coordinate[label=right:$v_{1}$]  (B) at (1,0);
      \coordinate[label=left:$v_{2}$]  (C) at (0,1);
      \coordinate[label=right:$v_{3}$] (D) at (1,1);

      \foreach \x in {A,B,C,D}
      \fill (\x) circle (2pt);

      \begin{scope}[every path/.style=solid,thick]
        \draw[->-] (A) -- (B);
        \draw[->-={.3}{}] (B) -- (C);
        \draw[->-] (C) -- (D);
        \draw[->-={.3}{}] (D) -- (A);
      \end{scope}
    \end{tikzpicture}
  \end{center}
  Then, $c$ is a simple cycle but $\gamma_c$ is not a simple closed curve since it has a self-intersection.
\end{remark}

\begin{remark}
  \label{remark:geometric}
  Under the condition that $\Sigma$ is a geometric simplicial complex,
  the $1$-simplices of a simple cycle $c$ can only intersect at its vertices.
  Furthermore, by definition the vertices of a simple cycle are pairwise distinct.
  In this case $\gamma_c$ is indeed a simple closed curve (in the sense of Remark~\ref{remark:simpleclosedcurve}).
\end{remark}

Next, a simple cycle $c$ is said to be \emph{planar} if its vertices $v_0,v_1,\hdots,v_{\ell-1}$ are coplanar.
Clearly, if $c$ is planar, the closed curve $\gamma_c$ it describes is also planar
(i.e.\ the image $\gamma([0,1])$ lies on the same plane).
A planar simple cycle $c$ is said to be \emph{convex}
if the following condition holds for $\gamma_c$:
for each line $L$, either
$L$ contains $\gamma_c(t)$  for at most two different values of $t \in [0,1)$,
or $L$ contains all values of $t$ within some interval in the circle $S^1 = [0,1]/\{0\sim 1\}$.

We state Fenchel's Theorem in the generality proven by Milnor~\cite{milnor1950total}
(the lower bound $2\pi \leq \kappa(c)$ and the equality condition below), applied to our setting.
The upper bound is $\kappa(c) \leq \ell \pi$ follows immediately from the definition of $\kappa(c)$.
\begin{theorem}[{Fenchel's Theorem \cite{fenchel1929krummung,borsuk1948courbure,milnor1950total}}]
  \label{thm:kappaineq}
  % Let $K$ be an abstract simplicial complex with vertex set $V \subset\bbR^N$.
  Let $c$ be a simple cycle of $K$ on $\ell$ vertices. Then,
  \[
    2\pi \leq \kappa(c) \leq \ell \pi,
  \]
  and $\kappa(c) = 2\pi$ if and only if $c$ is planar and convex.
\end{theorem}

\begin{remark}
  \label{remark:twopidegenerate}
  We illustrate the possibility of having $\kappa(c) = 2\pi$ for
  some simple cycle $c$
  with the closed curve $\gamma_c$ not simple
  in the sense of Remark~\ref{remark:simpleclosedcurve}.
  Consider four collinear vertices:
  \(
  \begin{tikzpicture}[point/.style={fill,
      shape=circle,inner sep=1pt,outer sep=0pt}]
    \coordinate[label=$v_{0}$]  (A) at (1,0);
    \coordinate[label=$v_{1}$]  (B) at (2,0);
    \coordinate[label=$v_{2}$]  (C) at (3,0);
    \coordinate[label=$v_{3}$]  (D) at (4,0);

    \foreach \x in {A,B,C,D}
    \fill (\x) circle (2pt) ;

    \begin{scope}[every path/.style=solid,thick]
      \draw[solid] (A) --(B);
      \draw[solid] (B) --(C);
      \draw[solid] (C) --(D);
      \draw[solid] (D) --(A);

      % \path pic["$\pi$",draw,angle radius=3mm,angle eccentricity=1.8] {angle = C--B--A};
      % \path pic["$\pi$",draw,angle radius=3mm,angle eccentricity=1.8] {angle = D--C--B};
    \end{scope}
  \end{tikzpicture}
  \)
  The simple cycle
  $
  c = [v_0,v_1] + [v_1,v_2] + [v_2,v_3] + [v_3,v_0]
  $
  has $\kappa(c) = \pi+ 0+0+\pi = 2\pi$,
  and is a convex curve (as can be checked from the definition given above).
\end{remark}

% \begin{theorem}[{Fenchel's theorem \cite{fenchel1929krummung,borsuk1948courbure,milnor1950total}}]
%   \label{thm:fenchel}
%   Let $\Sigma$ be a geometric simplicial complex with vertex set $V \subset\bbR^N$.
%   Let $c$ be a simple cycle of $\Sigma$ on $\ell$ vertices. Then,
%   $
%     2\pi \leq \kappa(c) \leq \ell \pi
%   $
%   with $2\pi = \kappa(c)$ if and only if then $c$ is convex planar.
% \end{theorem}

\begin{remark}
  Let us consider when the upper bound is achieved.
  Note first that for a simple cycle $\ell \geq 3$ is needed.
  With $\ell$ vertices, $\kappa(c) = \ell \pi$ implies that
  the exterior angle (turning angle) at each vertex of $c$ is $\pi$.
  This is possible, only for $\ell$ even  and the situation where
  all the vertices $v_i$ collinear and along their common line,
  each $v_i$ with $i$ even (respectively, odd)
  is to the left (respectively, right) of some point $p$
  (where the notion of left or right can be arbitrarily decided).

  For example, considering the situation of four collinear vertices
  \(
  \begin{tikzpicture}[point/.style={fill,
      shape=circle,inner sep=1pt,outer sep=0pt}]
    \coordinate[label=$w_{0}$]  (A) at (1,0);
    \coordinate[label=$w_{2}$]  (B) at (2,0);
    \coordinate[label=$w_{1}$]  (C) at (3,0);
    \coordinate[label=$w_{3}$]  (D) at (4,0);

    \foreach \x in {A,B,C,D}
    \fill (\x) circle (2pt) ;

    \begin{scope}[every path/.style=solid,thick]
      \draw[solid] (A) --(B);
      \draw[solid] (B) --(C);
      \draw[solid] (C) --(D);
      \draw[solid] (D) --(A);

      % \path pic["$\pi$",draw,angle radius=3mm,angle eccentricity=1.8] {angle = C--B--A};
      % \path pic["$\pi$",draw,angle radius=3mm,angle eccentricity=1.8] {angle = D--C--B};
    \end{scope}
  \end{tikzpicture}
  \),
  the simple cycle $c' = [w_0,w_1] + [w_1,w_2] + [w_2,w_3] + [w_3,w_0]$ has
  $\kappa(c') = 4 \pi$.
  Compare this with the cycle $c$ in Remark~\ref{remark:twopidegenerate}.
  Both $c$ and $c'$ are collinear. While $c$ (and $\gamma_c$) goes around once
  and is convex, $c'$ (and $\gamma_{c'}$) starting
  from $w_0$ goes to the right to $w_1$,
  turns around heading to $w_2$, turns around again to head to $w_3$,
  and then final turns around again to go back to $w_0$.
  We can see that $c'$ is not convex; for example, the vertical line that passes through the midpoint between $w_2$ and $w_1$ intersects $\gamma_{c'}$ exactly four times.
\end{remark}

\begin{remark}
  \label{remark:geometric2}
  Continuing Remark~\ref{remark:geometric},
  in the case that $c$ is a simple cycle in a geometric simplicial complex,
  such degeneracies do not occur. In this setting,
  if $\kappa(c) = 2\pi$, then $c$ describes a planar nondegenerate convex polygon.
\end{remark}

Finally, it is convenient to give the following definition,
which sets the minimum possible value at $0$.
\begin{definition}
  Define the \emph{reduced total absolute curvature} of a simple cycle $c$ to be
  \[
    \kappa'(c) := \kappa(c) - 2\pi.
  \]
\end{definition}

By Theorem~\ref{thm:kappaineq}, clearly
$0 \leq \kappa'(c) \leq (\ell-2)\pi$
for simple cycles $c$ on $\ell$ vertices.

%%% Local Variables:
%%% mode: LaTeX
%%% TeX-master: "main"
%%% End:

\section{Problem formulation}
\label{sec:formulation}

To formulate our optimization problem,
we consider the following generalization of the
total absolute curvature.

\begin{definition}[{cf.\ \cite{taniyama1998total}}]
  \label{defn:cycletac}
  Let $K$ be a simplicial complex with vertex set $V \subset \bbR^N$,
  and fix as basis for $C_k(K,\bbR)$
  the set of $k$-simplices with chosen orientations.
  \begin{enumerate}
  \item Denote by $\Zpm(K,\bbR)$ be the set of $z \in Z_1(K,\bbR) \subseteq C_1(K,\bbR)$
    with coefficients in $\{-1,0,1\}$ with respect to
    the basis of oriented $1$-simplices for $C_1(K,\bbR)$.
  \item For $z \in \Zpm(K,\bbR)$, the \emph{exterior angle} $\alpha(z,v)$ at a vertex
    $v \in V$ is
    the sum of
    $\theta(y-v,v-x)$ taken over
    $\{x,y\} \subset V$ such that $x \neq y$,
    $\{x,v\}$ and $\{y,v\}$ are both
    $1$-simplices of $K$, and % their corresponding oriented $1$-simplices
    both $[x,v]$ and $[y,v]$
    have nonzero coefficients in $z$.
    When there exist no such $\{x,y\} \subset V$ for $v$, $\alpha(z,v)$ is $0$.
    % with respect to
    % the basis of oriented $1$-simplices for $C_1(K,\bbR)$.
  \item The \emph{total absolute curvature} $\kappa(z)$ of $z \in \Zpm(K,\bbR)$
    is defined to be the total of its exterior angles
    \[
      \kappa(z) := \sum_{v\in V}\alpha(z,v).
    \]
  \end{enumerate}
\end{definition}

The exterior angle at $v$ is defined as
a sum over subsets $\{x,y\} \subset V$ satsfiying certain conditions.
By the fact that
$
\theta(y-v,v-x) = \theta(v-y, x-v) = \theta(x-v, v-y)
$
the summands are well-defined. Furthermore, it does not depend on the chosen orientations of the $1$-simplices.
The total absolute curvature defined here is the
total curvature for graphs of \cite{taniyama1998total} adapted to our setting.
It is also clear that for a simple cycle,
the definitions for the exterior angle in Definition~\ref{defn:simplecycleangles}
and in Definition~\ref{defn:cycletac} agree,
and likewise the definitions for the total absolute curvature in
Definition~\ref{defn:simplecycletac}
and in Definition~\ref{defn:cycletac} agree.

Recall that for each $k$, we choose as basis for $C_k(K,\bbR)$
the set of oriented $k$-simplices of $K$ with chosen orientations.
Furthermore, we fix the order of the basis elements.
With respect to this choice of ordered basis, this gives an isomorphism
$C_k(K,\bbR) \cong \bbR^{n_k}$ where $n_k$ is the number of $k$-simplices of $K$.
Below, for $x \in C_k(K,\bbR)$, we denote by $\mvec{x}$ the column vector of its coefficients
with respect to the chosen ordered basis for $C_k(K,\bbR)$.
We denote by $|\mvec{x}|$ the element-wise absolute value of $\mvec{x}$.

\begin{lemma}
  \label{lemma:quadraticform}
  Let $(e_1,e_2,\hdots, e_{n_1})$ be the chosen ordered basis of oriented $1$-simplices of $K$.
  Then, for $z \in \Zpm(K, \bbR)$,
  \begin{equation}
    \label{eqn:tac}
    \kappa(z) = \frac{1}{2} \abs{\vec{z}}^T Q \abs{\vec{z}}
  \end{equation}
  where $Q$ is the symmetric $n_1 \times n_1$ matrix with $(i,j)$th entry $q_{ij}$ defined as follows.
  The entry $q_{ij}$ is $\theta(y-v,v-x)$ if
  $i\neq j$ and $e_i, e_j$ share a vertex $v$ (and thus only the vertex $v$),
  in which case $y$ and $x$ are defined to be the vertices of $e_i$ and $e_j$ respectively
  distinct from $v$.
  Otherwise, $q_{ij}$ is $0$.
  % \[
  %   w_{ij} :=
  %   \begin{cases}
  %     \theta(p_{k}-p_{s},p_{l}-p_{s}) & e_i = \{p_k,p_s\},e_j = \{p_l,p_s\},\exists p_s\exists  p_k \neq p_l , \\
  %     0                                & otherwise
  %   \end{cases}
  % \]
\end{lemma}
We call $Q$ the \emph{exterior angle matrix} of $K$.
\begin{proof}
  We first check that $Q$ is indeed symmetric.
  The condition ``$i\neq j$ and $e_i, e_j$ share a vertex $v$'' is symmetric in $i$ and $j$.
  For the pairs $i,j$ where this condition holds,
  defining  $y$ and $x$ to be the vertices distinct from $v$ of $e_i$ and $e_j$ respectively, we have
  \[
    q_{ij} = \theta(y-v,v-x) = \theta(v-x,y-v) = \theta(x-v,v-y) = q_{ji}.
  \]
  For the pairs where the condition does not hold, $q_{ij} = 0 = q_{ji}$.

  Next, let us check Equation~\ref{eqn:tac}:
  \begin{align*}
    \frac{1}{2}\abs{\vec{z}}^T Q \abs{\vec{z}}
    & = \frac{1}{2}{\textstyle \sum_{
      \left\{i \neq j ~\middle|~ e_i, e_j \text{ both have nonzero coefficents in } z \right\}
      } q_{ij}}                                                      \\
    & = \frac{1}{2}{\textstyle \sum_{
      \left\{
      i \neq j ~\middle|~
      \substack{e_i, e_j \text{ both have nonzero coefficents in } z \\
        % e_i \neq e_j
    \text{and share exactly one vertex } v \text{ for some } v}
    \right\}} q_{ij}}                                              \\
    & = \frac{1}{2}
      \sum_{v \in V}
      {\textstyle \sum_{
      \left\{
      i \neq j ~\middle|~
      \substack{e_i, e_j \text{ both have nonzero coefficents in } z \\
    \text{and share only vertex } v }
    \right\}} q_{ij}}                                              \\
    & = \frac{1}{2}
      \sum_{v \in V} 2 \alpha(z,v) = \kappa(z)
  \end{align*}
  % where the first equality follows from the fact that $z \in \Zpm(K,\mathbb{R})$,
  % the second equality simply removes the cases where $e_i$ and $e_j$ share no vertices which implies that $q_{ij}=0$.
  where the coefficient of $2$ in the fourth equality
  comes from the fact that the inner summation
  adds up $q_{ij} = \theta(y-v,v-x)$
  taken over pairs of $e_i, e_j$ satisfying the given condition,
  whereas the summation in the definition for $\alpha(z,v)$
  is taken over (unordered) two-element subsets $\{x,y\} \subset V$.
\end{proof}

Next, interpret the total absolute curvature $\kappa(z)$ for $z\in \Zpm(K,\bbR)$
by connecting it with the total absolute curvature of simple cycles.
Recall that for a directed graph $G$, the outdegree (respectively, indegree) of vertex $v$,
denoted $\mathrm{outdeg}_G(v)$ (respectively, $\mathrm{indeg}_G(v)$)
is the number of directed edges with source (respectively, target) $v$.
We make the following observation.
\begin{lemma}
  \label{lem:zpm_directedgraph}
  There exists a bijection between nonzero $z \in \Zpm(K,\bbR)$
  and nonempty simple directed graphs $G$ satisfying the properties that:
  (i) its underlying undirected graph viewed as a simplicial complex is subcomplex of $K$
  and (ii) $\mathrm{outdeg}_G(v) = \mathrm{indeg}_G(v) > 0$ for each vertex $v$ of $G$.
\end{lemma}
\begin{proof}
  Given $G$ a directed graph satisfying the properties,
  let $z = \displaystyle\sum_{i \rightarrow j \text{: directed edge of } G} [i,j]$,
  which is nonzero because $G$ is nonempty.
  Then, we see that $z \in C_1(K, \bbR)$ by property (i),
  and $\partial_1(z) = 0$ by property (ii), showing that $z \in Z_1(K, \bbR)$.
  The coefficients of $z$ are in $\{-1,0,1\}$ by construction, and thus $z \in \Zpm(K, \bbR)$.

  In the other direction, by choosing orientations, each nonzero $z \in \Zpm(K,\bbR)$ can be written
  in the form  $z = \sum_{(i,j) \in K_0 \times K_0} z(i,j) [i,j]$
  where $z(i,i) = 0$ for all $i \in K_0$ and for each $i,j \in K_0$,
  either both $z(i,j)$ and $z(j,i)$ are $0$ or one of them is equal to $1$ and the other is $0$.
  Let $G'$ be the directed graph with vertices $K_0$ and directed edges $i\rightarrow j$ for $z(i,j)=1$.
  For each $v \in K_0$, the coefficient of $v$ in $\partial(z) = 0$ is given by:
  \[
    0 = \sum_{i: z(i,v)=1 } (-1)^0 + \sum_{j:z(v,j)=1} (-1)^1 = \mathrm{indeg}_{G'}(v) - \mathrm{outdeg}_{G'}(v).
  \]
  Remove isolated vertices from $G'$ to obtain $G$, which satisfies the properties (i) and (ii).
\end{proof}

Next, we recall the notion and existence of decomposition(s) into simple cycle(s), as follows.
\begin{lemma}
  \label{lemma:cycledecomposition}
  For nonzero $z \in \Zpm(K,\bbR)$, there exists a decomposition
  \begin{equation*}
    \label{eqn:cycledecomposition}
    z = \sum_{i=1}^m c_i
  \end{equation*}
  where $m\geq 1$, each $c_i$ is a simple cycle, and the $c_i$ are edge-disjoint.

  Furthermore, given such a decomposition, % in Equation~\eqref{eqn:cycledecomposition}
  the following hold.
  \begin{enumerate}
  \item $\vec{z} = \sum_{i=1}^m \vec{c_i}$, \label{lemitem:cycledecomposition_vec}
  \item \label{item:cycledecomposition_abssum} $\abs{\vec{z}} = \sum_{i=1}^m \abs{\vec{c_i}}$,
  \item $\abs{\vec{c_i}}^T \abs{\vec{c_j}} = 0$ for $i\neq j$,
  \end{enumerate}
  where the notation $\vec{x}$ means the column vector of coefficients of $x \in C_1(K,\bbR)$
  with respect to the chosen ordered basis of oriented $1$-simplices of $K$.
\end{lemma}
\begin{proof}
  By Lemma~\ref{lem:zpm_directedgraph}, we identify $z \in \Zpm(K,\bbR)$ with a simple directed graph $G$
  satisfying the condition $\mathrm{outdeg}_G(v) = \mathrm{indeg}_G(v)$ for each vertex $v$.
  It is easy to see that such a directed graph satisfying this condition
  partitions into edge-disjoint directed cycles (i.e.\ simple cycles $c_i$ viewed as directed graphs)\footnote{This can be shown by a simple modification of the proof for the well-known statement for undirected graphs with each vertex having even degree (Veblen's theorem \cite{veblen1912application}; see also for example \cite[Theorem~I.1]{bollobas1998modern}).}.
  This gives the corresponding decomposition in
  \eqref{eqn:cycledecomposition}. Item~\ref{lemitem:cycledecomposition_vec} follows from linearity of taking coefficients with respect to basis. The remaining two items follow from
  the fact that $c_i$ are edge-disjoint.
\end{proof}

\begin{theorem}
  \label{thm:decompo}
  Let $z = \sum_{i=1}^m c_i$ be a decomposition of a nonzero $z \in \Zpm(K,\bbR)$
  into pairwise edge-disjoint simple cycles $c_i$ as in Lemma~\ref{lemma:cycledecomposition}.
  Then
  \[
    \kappa(z)
    = \sum_{i=1}^m \kappa(c_i) + \frac{1}{2}\sum_{i\neq j} \abs{\vec{c_i}}^T Q \abs{\vec{c_j}}
    = \sum_{i=1}^m \kappa'(c_i) + 2\pi m + \frac{1}{2}\sum_{i\neq j} \abs{\vec{c_i}}^T Q \abs{\vec{c_j}}.
  \]
\end{theorem}
\begin{proof}
  % By
  This immediately follows from Lemma~\ref{lemma:quadraticform},
  Lemma~\ref{lemma:cycledecomposition} item~\ref{item:cycledecomposition_abssum},
  and the definition of the reduced total absolute curvature $\kappa'$.
\end{proof}
This shows that the total absolute curvature $\kappa(z)$ decomposes
into a sum of the total absolute curvatures of the simple cycles in the decomposition
plus a term penalizing shared vertices between different simple cycles.
Alternatively we can express it as a
sum of the reduced total absolute curvatures, a term penalizing the number of simple cycles,
and a term penalizing shared vertices.

\begin{corollary}
  Let $z \in \Zpm(K,\bbR)$ such that there exists a decomposition of $z$ into $m$ pairwise edge-disjoint simple cycles.
  Then
  \[
    \kappa(z) \geq 2\pi m.
  \]
\end{corollary}
\begin{proof}
  This follows from Theorem~\ref{thm:decompo} and the facts that
  $\kappa(c_i) \geq 2\pi$ for each simple cycle $c_i$ by Theorem~\ref{thm:kappaineq}, and
  that $\sum_{i\neq j} \abs{\vec{c_i}}^T Q \abs{\vec{c_j}} \geq 0$
  because the entries of $Q$ are all non-negative.
\end{proof}

Next, we formulate the angle-optimal homologous cycle problem (AOHCP) with $\kappa(z)$ as the objective function.
Given $z_0 \in \Zpm(K, \bbR)$,
\begin{mini}
  {}{\kappa(z) = \frac{1}{2} \abs{\vec{z}}^T Q \abs{\vec{z}}}
  {}{}
  \addConstraint{z \sim z_0 ~(\text{i.e. } z - z_0 \in B_1(K, \bbR))}
  \addConstraint{z \in \Zpm(K, \bbR)}.
\end{mini}
To improve the interpretability of the feasible solutions
and to have all variables binary in the problem \eqref{opt:BQP} below,
we further restrict the problem to require
\[
  z - z_0 \in \Bpm(K, \bbR)
\]
where $\Bpm(K, \bbR)$ is the set of $y \in B_1(K,\bbR)$
with coefficients in $\{-1,0,1\}$ with respect to the basis of oriented $1$-simplices.

Recall the well-known transformation
of $\vec{x} = \vec{x}^{+} - \vec{x}^{-}$
and $\abs{\vec{x}} = \vec{x}^{+} + \vec{x}^{-}$
where $\vec{x}^{+}, \vec{x}^{-} \geq 0$ are the positive and negative parts of a vector $\vec{x}$.
Using this transformation, we reformulate the optimization problem as
\begin{mini}
  {}{\frac{1}{2} (\vec{x}^{+}+\vec{x}^{-})^{\top}Q(\vec{x}^{+}+\vec{x}^{-})}
  {\label{opt:BQP}}{}
  \addConstraint{(\vec{x}^{+}-\vec{x}^{-}) - \vec{z_0} = [\partial_2](\vec{y}^{+}-\vec{y}^{-})}{}
  \addConstraint{\vec{x}^{+},\vec{x}^{-} \in \{0,1\}^{n_1}}
  \addConstraint{\vec{y}^{+},\vec{y}^{-} \in \{0,1\}^{n_2}}
\end{mini}
where $[\partial_2]$ is the matrix of $\partial_2$ with respect to the chosen bases of oriented $2$-simplices and $1$-simplices.
This can be rewritten in the standard form of a binary quadratic programming problem:
\begin{mini}
  {}{\frac{1}{2} \vec{w}^{\top}\hat{Q}\vec{w}}
  {\label{opt:StandardBQP}}{}
  \addConstraint{A\vec{w} = \vec{z_0}}{}
  \addConstraint{\vec{w} \in \{0,1\}^{2 ({n_1}+{n_2})}}
\end{mini}
where
\[
  \vec{w} =
  \begin{bmatrix}
    \vec{x}^{+} \\ \vec{x}^{-} \\ \vec{y}^{+} \\ \vec{y}^{-}
  \end{bmatrix}
  ,~
  A =
  \begin{bmatrix}
    I & -I & -[\partial_2] & [\partial_2]
  \end{bmatrix}
  \text{, and }
  \hat{Q} =
  \begin{bmatrix}
    Q & Q & O & O \\
    Q & Q & O & O \\
    O & O & O & O \\
    O & O & O & O \\
  \end{bmatrix}
\]
where $I$ is the $n_1 \times n_1$ identity matrix and $O$ are zero matrices of appropriate sizes.
This is a binary quadratic programming problem with
$2(n_1+n_2)$ binary variables and $n_1$ linear equality constraints.
Then, given an optimal solution
$
\vec{w}^\ast =
\begin{bmatrix}
  (\vec{x}^{+})^T & (\vec{x}^{-})^T & (\vec{y}^{+})^T & (\vec{y}^{-})^T
\end{bmatrix}^T
$
we reconstruct $z \in \Zpm(K, \bbR)$ by setting
$\vec{x} := \vec{x}^{+} - \vec{x}^{-}$ and letting $z$ be
the cycle such that the column vector of its coefficients
with respect to the chosen ordered basis for $C_1(K,\bbR)$ is $\vec{x}$.
Before making additional observations, we outline
the basic procedure for solving the AOHCP in Algorithm~\ref{alg:AngleOptimizeCycle}.

\begin{algorithm}[H]
  \caption{Basic procedure for solving the AOHCP}
  \label{alg:AngleOptimizeCycle}
  \begin{algorithmic}
    \Require A simplicial complex $K$ with vertices in $\bbR^N$ and $z_0 \in \Zpm(K, \bbR)$
    \Procedure{AngleOptimizeCycle}{$K$, $z_0$}
      \State Calculate the exterior angle matrix $Q$ and $2$nd boundary matrix $[\partial_2]$ of $K$.
      % \State Calculate $\hat{Q}$ and $A$ from $W$, $B$, $[\partial_2]$.
      \State Construct and solve the problem \eqref{opt:StandardBQP} (or its reformulation), and
      \State {~~~~} obtain an optimal solution $\vec{w}^\ast$.
      \State Construct $z^\ast \in \Zpm(K, \bbR)$ from $\vec{w}^\ast$.
      \State \Return $z^\ast$ and its objective value
    \EndProcedure
  \end{algorithmic}
\end{algorithm}

Note that $\hat{Q}$ is symmetric matrix with nonnegative entries and with zero diagonal.
We observe that $\hat{Q}$ cannot be positive semidefinite unless it is the zero matrix,
and thus the quadratic form
$f(\vec{w}) := \frac{1}{2} \vec{w}^{\top}\hat{Q}\vec{w}$
is not a convex function in general.

There are various techniques and reformulations for dealing with binary quadratic programming problems,
some of which we note below.
For example, we can rewrite the problem as an equivalent mixed-integer programming problem
using the standard linearization by Glover and Woolsey \cite{glover1974converting}
(together with an observation in \cite{forrester2008quadratic}
to omit some of the redundant additional constraints)
to obtain the problem
\begin{mini}
  {}{\sum_{(i,j) \in D}{\hat{q}_{ij}}\alpha_{ij}}
  {\label{opt:MILP}}{}
  \addConstraint{A\vec{w} = \vec{z_0}}{}
  \addConstraint{\vec{w} \in \{0,1\}^{2 ({n_1}+{n_2})}}
  \addConstraint{\alpha_{ij} \ge w_i + w_j -1 ~~ \text{for } (i,j) \in D}
  \addConstraint{\alpha_{ij} \ge 0 \qquad\qquad\quad~ \text{for } (i,j) \in D}
\end{mini}
where $D = \{(i,j) \mid 1 \leq i < j \leq 2{n_1}, \hat{q}_{ij} \neq 0\}$.
% \begin{mini}
%   {}{\sum_{i=0}^{2{n_1}-1}\sum_{j=i+1}^{2{n_1}-1}{\hat{q}_{ij}}\alpha_{ij}}
%   {\label{opt:MILP}}{}
%   \addConstraint{A\vec{w} = \vec{z_0}}{}
%   \addConstraint{\vec{w} \in \{0,1\}^{2 ({n_1}+{n_2})}}
%   \addConstraint{\alpha_{ij} \ge w_i + w_j -1 \ (0 \le i<j \le 2{n_1}-1)}
%   \addConstraint{\alpha_{ij} \ge 0 \ (0 \le i<j \le 2{n_1}-1).}
% \end{mini}
%
The linearization increases both the number of variables and the number of constraints in the optimization problem.
In particular, it introduces
$|D|$
% $n_1(2n_1 - 1)$
additional variables $\alpha_{ij}$ ($(i,j) \in D$) and
$2|D|$
% $2n_1(2n_1 - 1)$
additional inequality constraints.
Although the problem size increases,
it becomes possible to use standard mixed-integer programming solvers.
See \cite{adams2007linear,furini2019theoretical} and the references therein, for example,
for more recent developments in linearization techniques.

% Next, we will explain the algorithm for solving the optimization program. The algorithm is as follows. We call the solution for the algorithm angle optimal cycle.

% \begin{algorithm}[H]
%   \caption{Algorithm for an angle-optimal cycle}
%   \label{AngleOptimizeCycle}
%   \begin{algorithmic}
%     \Require a simplicial complex $K$ in $\mathbb{R}^N$ and a 1-dimensional cycle $z_0$
%     \State \textbf{Procedure:} AngleOptimizeCycle($K$,$z_0$)
%     \State \textbf{Step$1$:}
%     calculate $W,B,[\partial_2]$ from $K$
%     \State \textbf{Step$2$:}
%     calculate $\hat{Q}$ and $A$ from $W,B,[\partial_2]$
%     \State \textbf{Step$3$:}
%     construct optimization problem and solve the solution
%     \State
%     \begin{mini}
%       {}{\sum_{i=0}^{2{n_1}-1}\sum_{j=i+1}^{2{n_1}-1}{\hat{q}_{ij}}\alpha_{ij}}
%       {\tag{\ref{opt:MILP}}}{}
%       \addConstraint{A\vec{w} = \vec{z_0}}{}
%       \addConstraint{\vec{w} \in \{0,1\}^{2 ({n_1}+{n_2})}}
%       \addConstraint{\alpha_{ij} \ge w_i + w_j -1 \ (0 \le i<j \le 2{n_1}-1)}
%       \addConstraint{\alpha_{ij} \ge 0 \ (0 \le i<j \le 2{n_1}-1)}
%     \end{mini}
%     \Return angle optimal cycle and its objective value
%   \end{algorithmic}
% \end{algorithm}

As the size of the problem (i.e.\ the numbers
$n_1$ and $n_2$
of the
$1$-simplices and $2$-simplices respectively in $K$)
increases,
the time needed to solve the optimization problem may increase dramatically.
%
% Therefore, we explain a method to reduce computation time and compute a constrained optimal solution.
The following heuristics can be applied.
The first idea, imitating one of the heuristics proposed in \cite[Section~4.2]{obayashi2018volume},
is to solve the problem within a smaller simplicial complex $K'$
that is an appropriate neighborhood of $z_0$,
instead of using the original simplicial complex $K$. For example, one can choose $K'$ to be
the subcomplex of $K$ induced by vertices of $K$ with distance at most $\epsilon$ to some vertex of $z_i$.
Then, we can automatically increase $\epsilon$ as needed when an optimal solution is not found.
In conjuction, we can also apply the heuristic iteratively,
that is, iteratively solving the problem with $K_i$ an appropriate neighborhood of the $z_i$ for $i=0,1,\hdots$ and letting
$z_{i+1}$ be the result of $\texttt{AngleOptimizeCycle}(K_i, z_i)$.

%%% Local Variables:
%%% mode: LaTeX
%%% TeX-master: "main"
%%% End:

\section{Computational Demonstrations}
\label{sec:computations}

We show some computational experiments demonstrating the results of applying the proposed
Algorithm~\ref{alg:AngleOptimizeCycle} to point clouds $V \subset \mathbb{R}^{3}$.
We specify how we choose the simplicial complex $K$ and the cycle
$z_0 \in \Zpm(K, \bbR)$
as inputs to Algorithm~\ref{alg:AngleOptimizeCycle}.

We generate some point clouds $V \subset \mathbb{R}^3$,
construct its alpha complex \cite{edelsbrunner1994three} filtration
$\{\mathrm{Alpha}_r(V)\}_{r \in \mathbb{R}}$,
and compute its dimension $1$ persistence diagram \cite{edelsbrunner2002topological} $D_1$.
Then, we choose a birth-death pair $[b,d) \in D_1$ with the longest lifespan $d-b$,
and let $z_0$ be its representative cycle.
Finally, we choose some $t \in [0,1)$
and choose the simplicial complex $K_r := \mathrm{Alpha}_r(V)$ with
$r := b + t(d-b)$ as our simplicial complex $K$.
Note that $z_0$ is born at $b$, and thus $z_0 \in Z_1(K_b, \bbR) \subseteq Z_1(K_r, \bbR)$
can indeed be considered as a $1$-cycle of $K_r$.
In general, it is possible that $z_0 \notin \Zpm(K_r, \bbR)$
(i.e.\ some of the coefficients of the simplices in the representative cycle $z_0$ are not in $\{-1,0,1\}$).

\begin{remark}
For a choice of $r$ with $K_r \supsetneq K_b$, a computed optimal solution $z^*$ by Algorithm~\ref{alg:AngleOptimizeCycle}
may involve simplices born after $b$ and thus is not a cycle in $K_b$. In particular, $z^*$ may not be a representative cycle for $[b,d)$.
However, it is by construction homologous to $z_0$ in $K_r$. It is in this sense that we consider the optimization problem.
Considering the problem for persistence representatives is potential future work.
\end{remark}

In practice, a possible pipeline integrating Algorithm~\ref{alg:AngleOptimizeCycle}
into a persistent homology analysis would adapt a similar pipeline, where one can use
some other filtration $\{K_r\}_{r \in \bbR}$ instead of the alpha complex filtration and/or
some other birth-death pair $[b,d)$ (not necessarily the one with the longest lifespan) and its associated representative cycle $z_0$.

For our experiments, we prepare several point clouds $V \subseteq \bbR^3$, summarized in Table~\ref{table:inputs}.
The data of type ``cylinder'' is created by randomly sampling ($n_0 = 300, 500, 1000$, respectively) points from the surface of a hollow cylinder with radius $1$ and height $2$.
The data of type ``slipper'' is a point cloud shaped like slipper but with no sole.
The base (boundary of the sole) of the slipper is an ellipse with major axis of length $4$ and minor axis of length $2$.
The toe cap and upper vamp of the slipper is formed as a surface of (half) revolution
by rotating half of the front part of the base. This is a portion of an ellipsoid.
For the slipper data, the points are not randomly sampled but instead obtained by taking regularly-spaced angles
in spherical coordinate system.

\begin{table}[h]
  \caption{Summary of the data used.
    We list the number of vertices $n_0$ and some details about the initial cycle $z_0$ representing the longest persistence interval $[b,d)$:
    $|z_0|$ is the number of $1$-simplices of $z_0$,
    $\ell(z_0)$ is the length (the sum of the lengths of the $1$-simplices) of $z_0$,
    and $\kappa(z_0)$ is its total absolute curvature.}
  \label{table:inputs}
  \begin{tabular}{|c|c|c|c|c|c|c|}
    \hline
       & Data type & $n_0$ & $[b,d)$                                     & $|z_0|$ & $\ell(z_0)$              & $\kappa(z_0)$  \\ \hline\hline
    C1 & cylinder  & $300$ & $[\num{0.02322393425}, \num{1.000000374})$ & $74$ & $\num{10.200768097054576}$ & $\num{21.247393198983932} \pi$ \\
    C2 & cylinder  & $500$ & $[\num{0.01137901322}, \num{1.00000004})$ & $107$ & $\num{11.230938502846687}$ & $\num{26.339357944443133} \pi$ \\
    C3 & cylinder  & $1000$ & $[\num{0.00570843568}, \num{1.000000002})$ & $185$    & \num{13.951517522660193} & $\num{47.5274394585}\pi$ \\
    S1 & slipper   & $201$ & $[\num{0.02713788266}, \num{0.9417295932})$ & $52$ & $\num{9.349799958023352}$ & $\num{4.511300712736921} \pi$ \\
    S2 & slipper   & $601$ & $[\num{0.01169831303}, \num{0.9376226229})$ & $90$ & $\num{8.893500705825458}$ & $\num{4.700381723548458} \pi$ \\
    \hline
  \end{tabular}
\end{table}

% Note: the initial cycle "cyclebefore.pdf" does not vary based on choices of timelimit sec and ratio
% Thus, it should not matter which subfolder we choose for the following set of subfigures. For consistency, I choose the one with 0.1ratio across all data.
\begin{figure}[H]
  \begin{subfigure}[]{0.3\linewidth}
    \includegraphics[width=\columnwidth]{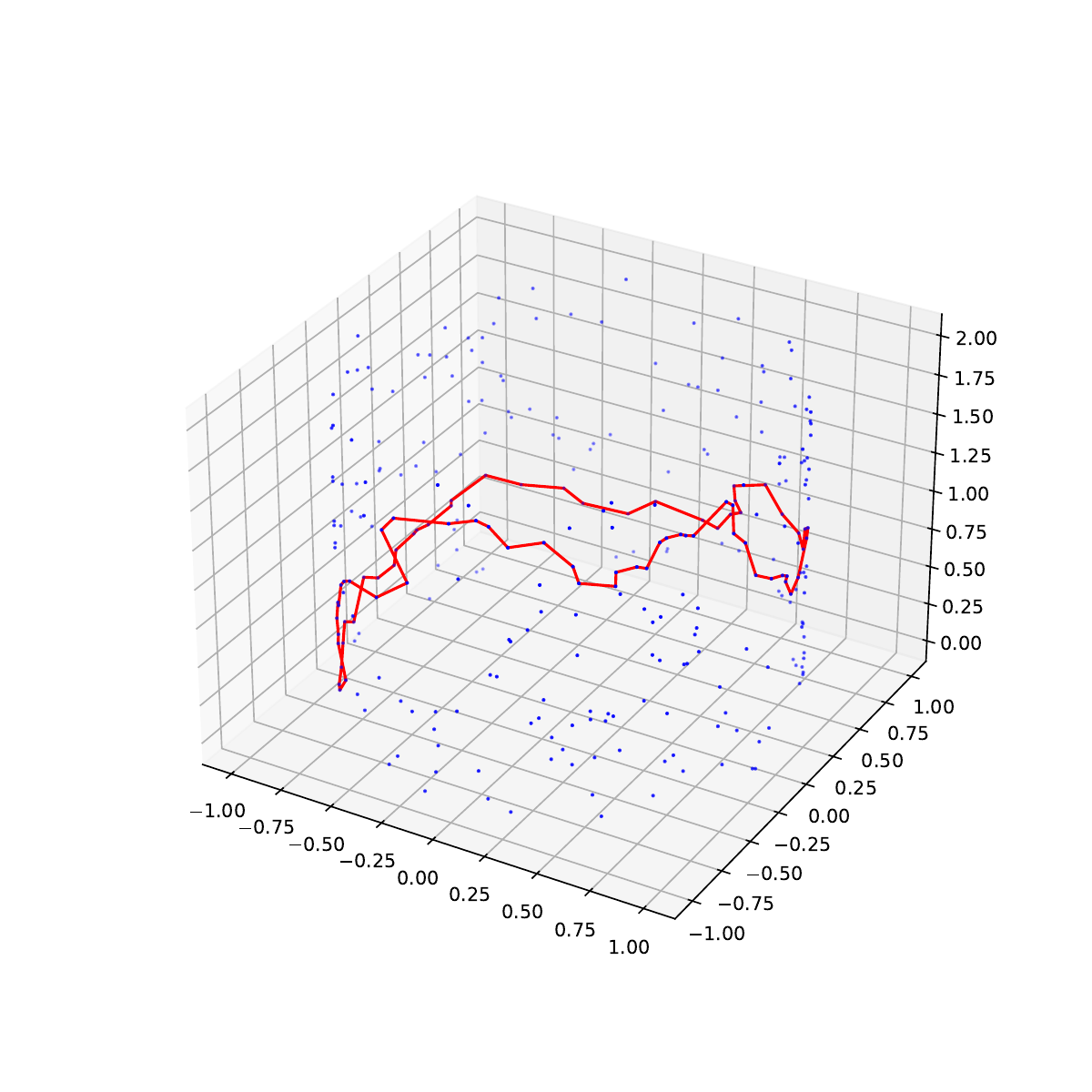}
    \caption{C1, $\kappa(z_0) \approx \num{21.247393198983932} \pi$}
  \end{subfigure}
  \begin{subfigure}[]{0.3\linewidth}
    \includegraphics[width=\columnwidth]{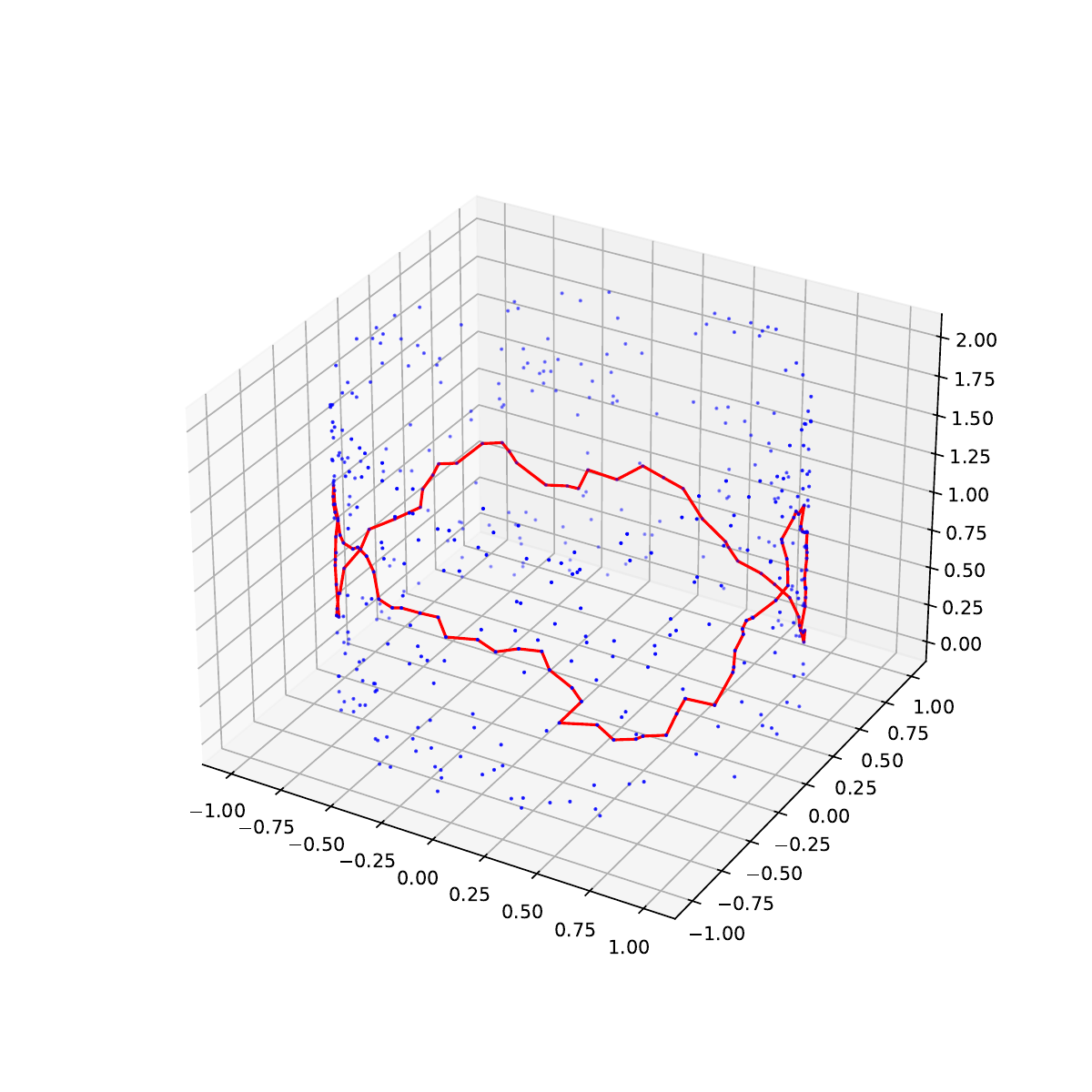}
    \caption{C2, $\kappa(z_0) \approx \num{26.339357944443133} \pi$}
  \end{subfigure}
  \begin{subfigure}[]{0.3\linewidth}
    \includegraphics[width=\columnwidth]{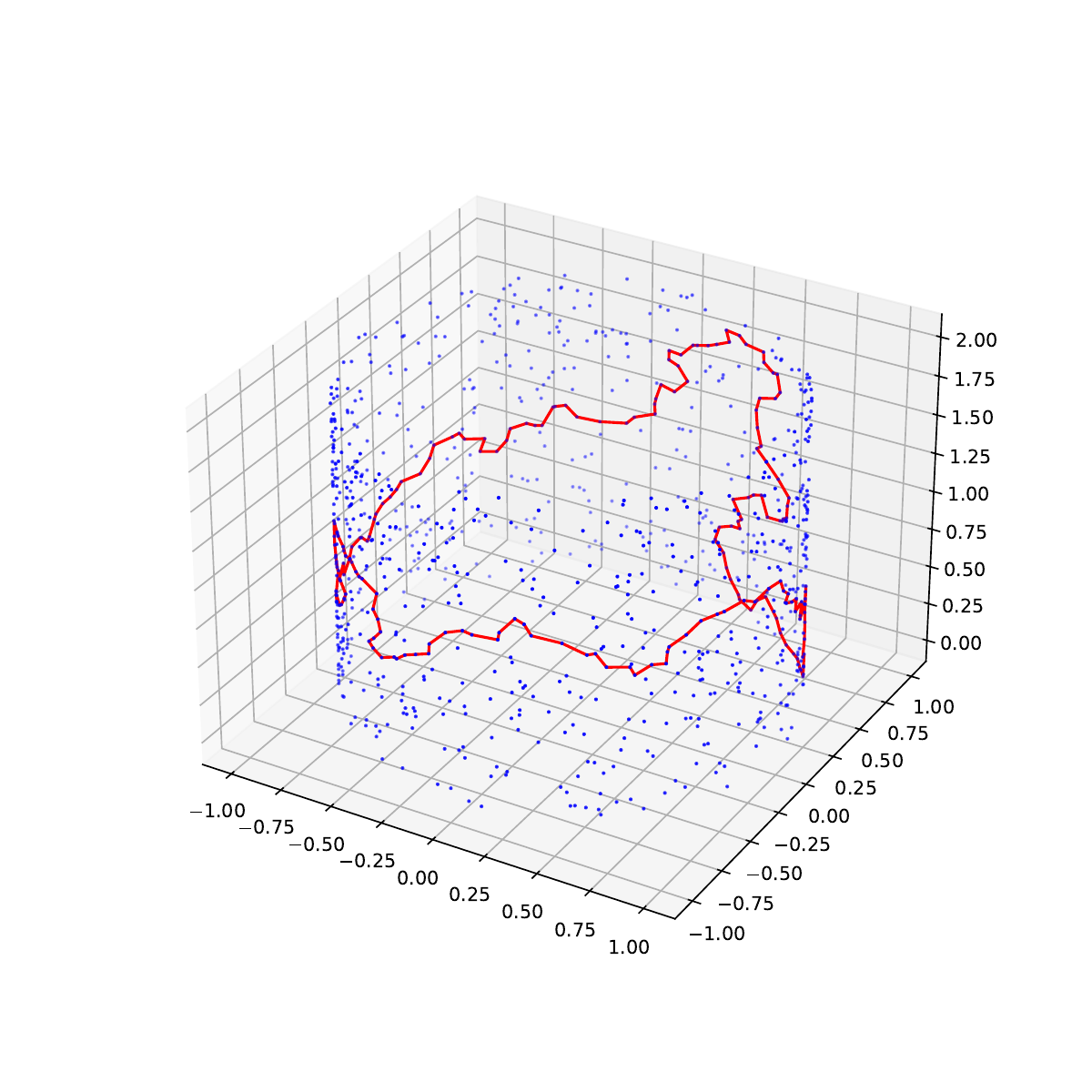}
    \caption{C3, $\kappa(z_0) \approx \num{47.5274394585}\pi$}
  \end{subfigure}
  \\
  \centering
  \begin{subfigure}[]{0.3\linewidth}
    \includegraphics[width=\columnwidth]{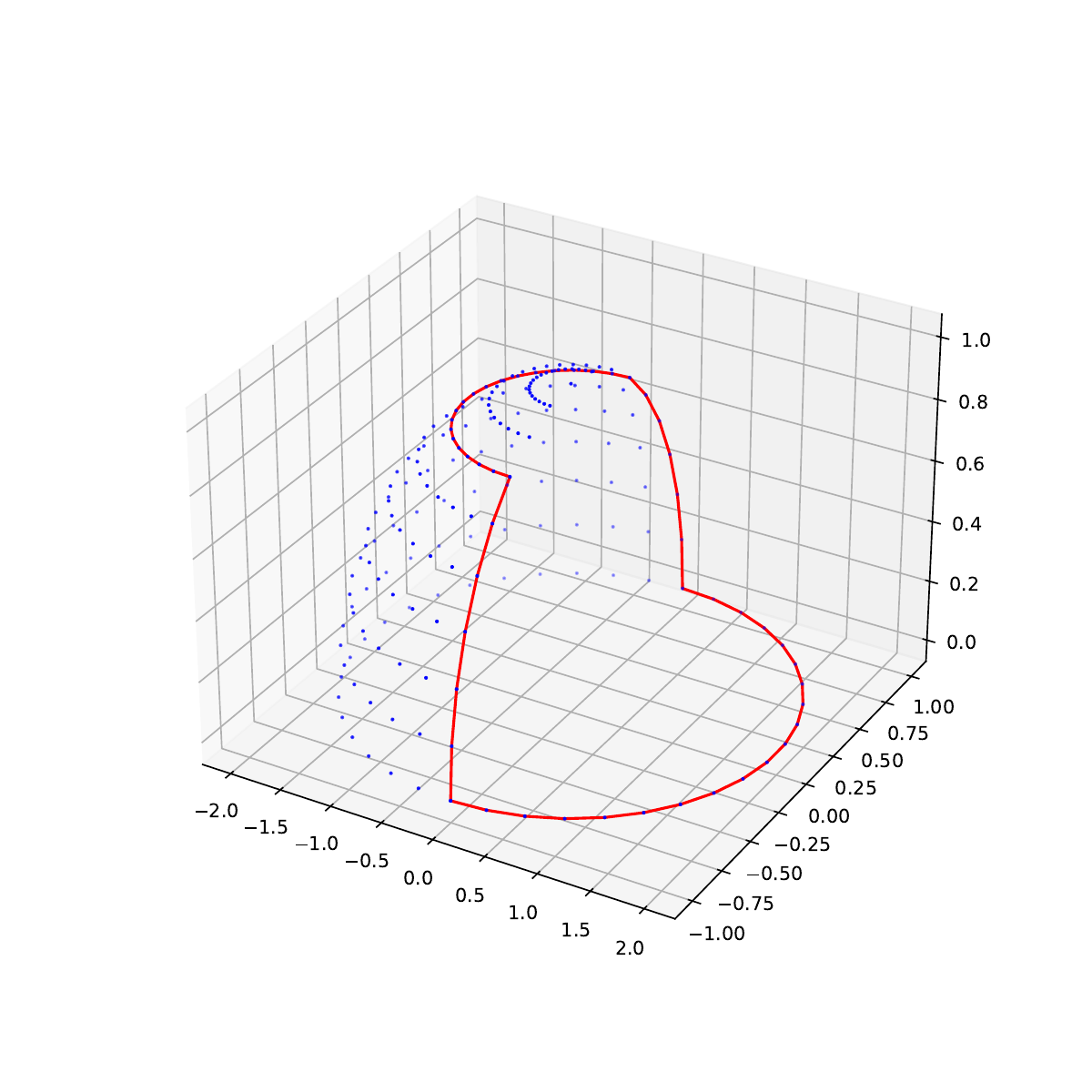}
    \caption{S1, $\kappa(z_0) \approx \num{4.511300712736921} \pi$}
  \end{subfigure}
  \begin{subfigure}[]{0.3\linewidth}
    \includegraphics[width=\columnwidth]{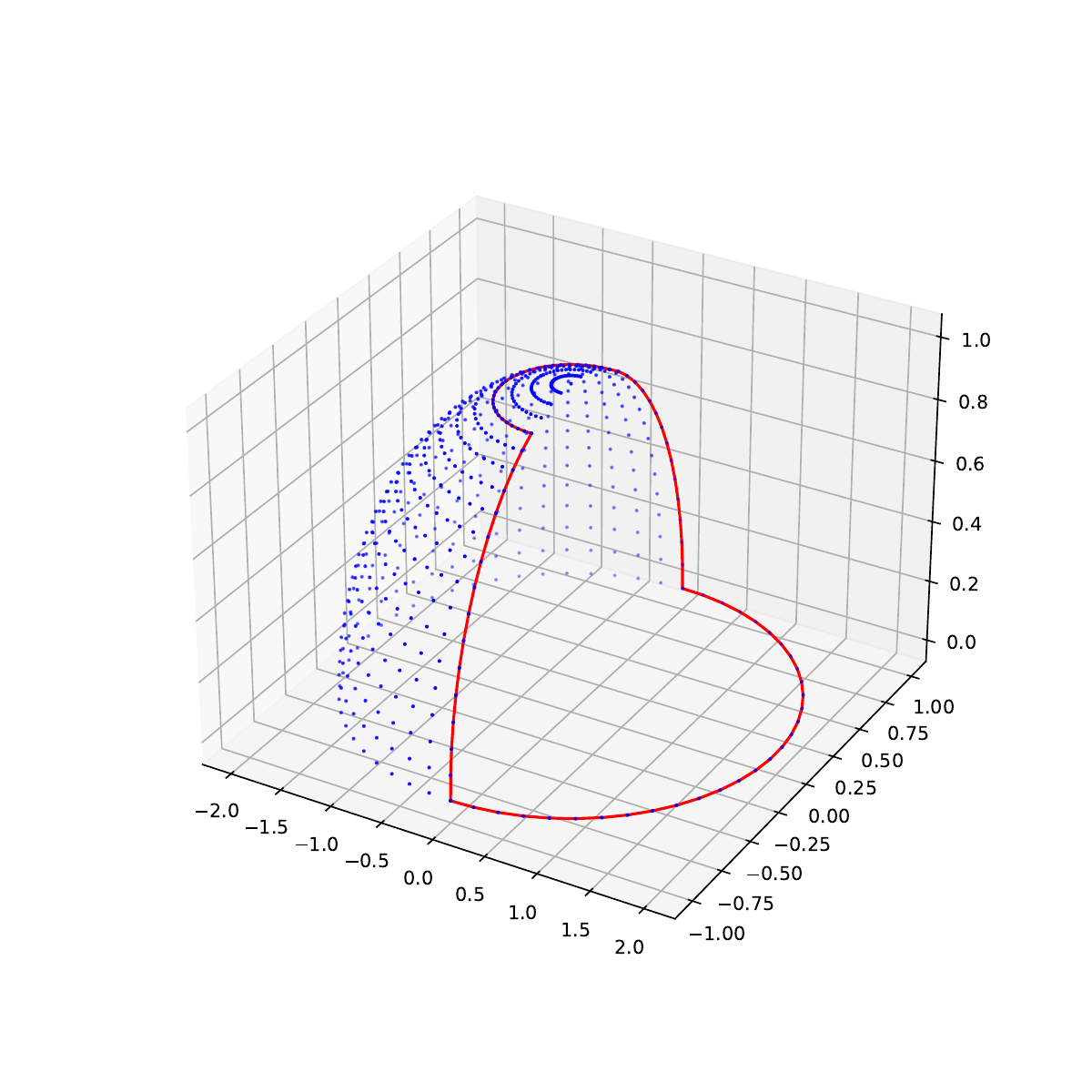}
    \caption{S2, $\kappa(z_0) \approx \num{4.700381723548458} \pi$}
  \end{subfigure}
  \caption{Initial cycles $z_0$ and their total absolute curvatures $\kappa(z_0)$.}
  \label{fig:init}
\end{figure}

For our experiments, we used OptiPersLP~\cite{OptiPersLP}
to compute persistent homology and representative cycles\footnote{OptiPersLP uses CGAL~\cite{cgal} to compute the alpha filtration. The main feature of OptiPersLP is to compute optimal cycle representatives according to the algorithm described in a prior work~\cite{escolar2016ocphlp} which focuses on minimizing the number of simplices in the cycle. For our purposes, we bypass (disable) this optimization since we want to perform the minimization of the total absolute curvature.}. We obtain $K_r$ and $z_0 \in \Zpm(K, \bbR)$ as described above.
For solving the optimization problem in Algorithm~\ref{alg:AngleOptimizeCycle}, we use the
software {IBM ILOG CPLEX Optimization Studio 22.1.1} \cite{Cplex} via
the DOcplex \cite{docplex} library in python.
Experiments were performed on a MacBook Pro (2024 M4 Pro, 48GB RAM).

For the optimization problem, due to the nature of the problem, we do not let the solver run until an optimal solution is found.
Instead, by setting the \verb|timelimit| parameter in DOcplex, we periodically pause\footnote{
  See \url{https://www.ibm.com/docs/en/icos/22.1.1?topic=parameters-optimizer-time-limit-in-seconds} for documentation on the \verb|timelimit| parameter.
  Since the time is measured in terms of wall clock time, variations in measured time may cause the solver to hit the time limit at different phases of the solving process, leading to non-determinism.}
and resume the optimization process.
At the $i$th pause we capture the best solution found so far as $z_i$\footnote{This $z_i$ is the best solution found after the $i$th pause, and is not related to the $z_i$ solution in the iterative heuristic discussed at the end of the previous Section~\ref{sec:formulation}.}. For each point cloud data, (by experimentation) we appropriately choose a value for \verb|timelimit|.
The details of these solutions are displayed Tables~\ref{table:cylinder},~\ref{table:cylinder_more},~\ref{table:slipper_S1},~\ref{table:slipper_S2}.

Theorem~\ref{thm:kappaineq} gives the general lower bound of $2\pi$ for the total absolute curvature $\kappa(z)$ of simple cycles $z$.
However, whether or not the current $K_r$ under consideration supports such a simple cycle is not guaranteed in general,
and thus for particular $K_r$ the global minimum value for $\kappa(z)$ is not necessarily $2\pi$.
It does, however, serve as a useful reference point for understand the computed values of $\kappa$.
For example, in Figure~\ref{fig:init},
we see that larger values of $\kappa(z_0)$ correspond to more visually ``jagged'' cycles.

\subsection{Results on the ``cylinder'' data}

\FloatBarrier

In Table~\ref{table:cylinder}, we display the results for running the optimizer on the cylinder data C1 and C2,
where the three subrows within each correspond to the choices of $t = 0.1$, $0.2$, and $0.4$
for defining $r := b + t(d-b)$.
In general, increasing $r$ leads to more simplices in the simplicial complex $K_r$, leading to a larger optimization problem.
We also tabulate the corresponding number of $i$-simplices $n_i$ ($i\in\{1,2\}$)
and recall that our binary quadratic programming problem \eqref{opt:StandardBQP} has
$2(n_1+n_2)$ binary variables and $n_1$ constraints.
For the results Table~\ref{table:cylinder},
the solver is repeatedly paused and restarted with a time limit of $5$ seconds before pausing.
At the $i$th pause we capture the best solution found so far as $z_i$.
We defer the results for C3 to Table~\ref{table:cylinder_more},
as it required longer solver time limits to get reasonable-looking solutions (see Figure~\ref{fig:C3_more}).

In addition to the total absolute curvature $\kappa(z_i)$, we also list its number $|z_i|$ of $1$-simplices
and its length $\ell(z_i)$ (sum of the lengths of its $1$-simplices), for reference.
Tautologically (by definition), running the solver can lead to better solutions with $\kappa(z_1) \geq \kappa(z_2) \geq \kappa(z_3) \geq \hdots$.
We see this in general, going from left to right in Table~\ref{table:cylinder}.
However, depending on the size of the data, we see that for the larger point clouds we
see minimal improvements in $\kappa$ given the relatively short time limit.

We illustrate some of original representative cycles and computed solutions
in Figure~\ref{fig:C1r1}~and~\ref{fig:C2r3}.
For these runs, we see that after $z_1$ we see minimal improvements in $\kappa$ over the next $10$ seconds.
In Table~\ref{table:cylinder_more} and Figure~\ref{fig:C2_more} we display the result for a longer time limit.
% intuitively it should be possible to achieve smaller values of $\kappa$.

We also note that in Table~\ref{table:cylinder}, going from left to right
with decreasing $\kappa$, the values of $|z_i|$ and $\ell(z_i)$ are also decreasing,
even though they are not being minimized for directly in Algorithim~\ref{alg:AngleOptimizeCycle}.
This makes intuitive sense; for this cylinder data,
among cycles that wrap around the hole, any cycle that goes around the cylinder parallel to the base
would be planar and at the same time have smallest length and smallest number of $1$-simplices.
In such a situation, solving a standard optimal homologous cycle problem
minimizing length
would be more efficient since it can be cast as a comparatively smaller linear optimization problem \cite{dey2010optimal,escolar2014computing}.
We shall see via the next example that
minimizing total absolute curvature is not the same as minimizing length,
in general.

% *************** START CYLINDER FLOATS ***************

\begin{table}[h]
  \caption{Results for the ``cylinder'' data C1 and C2.
    The solver is repeatedly paused and restarted with a time limit of $5$ seconds before pausing.
    At the $i$th pause we capture the best solution found so far as $z_i$.}
  \label{table:cylinder}
    \begin{tabular}{|c|c|c|c| |c|c|c| |c|c|c| |c|c|c|}
    \hline
  & & \multicolumn{2}{c||}{size of $K_r$}  & \multicolumn{3}{c||}{$i=1$ ($5$ sec)} & \multicolumn{3}{c||}{$i=2$ ($10$ sec)} & \multicolumn{3}{c|}{$i=3$ ($15$ sec)} \\
  \hline
  % & $r$
  & $t$
 & $n_1$ & $n_2$ & $|z_1|$   & $\ell(z_1)$ & $\kappa(z_1)$ & $|z_2|$   & $\ell(z_2)$ & $\kappa(z_2)$ & $|z_3|$   & $\ell(z_3)$ & $\kappa(z_3)$ \\ \hline\hline
  \multirow{3}{*}{C1}
  & 0.1
  & $945$ & $735$
 & $29$ & $\num{7.02055957260189}$ & $\num{5.661044433980622} \pi$
 & $27$ & $\num{6.636087858464085}$ & $\num{4.235763078989115} \pi$
 & $26$ & $\num{6.452537302129011}$ & $\num{3.3733039691048523} \pi$ \\
  % & \num{0.212678047912}
  & 0.2
  & $1005$ & $847$
 & $32$ & $\num{8.474476607642929}$ & $\num{7.193848201355509} \pi$
 & $32$ & $\num{8.474476607642929}$ & $\num{7.193848201355509} \pi$
 & $28$ & $\num{7.724042884772975}$ & $\num{5.511341021323253} \pi$ \\
  % & \num{0.40950855543400005}
  & 0.4
  & $1098$ & $1025$
 & $18$ & $\num{6.465860972031913}$ & $\num{3.2579998761894955} \pi$
 & $18$ & $\num{6.465860972031913}$ & $\num{3.2579998761894955} \pi$
 & $17$ & $\num{6.399339600485485}$ & $\num{2.853520508460676} \pi$ \\ \hline
  \multirow{3}{*}{C2}
  % & \num{0.10848945834160001}
  & 0.1
  & $1588$ & $1229$
 & $38$ & $\num{7.241315195289824}$ & $\num{8.56997168833635} \pi$
 & $37$ & $\num{7.051743849187168}$ & $\num{7.857805955441441} \pi$
 & $34$ & $\num{6.797085144771437}$ & $\num{5.762065998690735} \pi$
  \\
  % & \num{0.20754618585920004}
  & 0.2
  & $1677$ & $1402$
 & $47$ & $\num{9.903508341092198}$ & $\num{12.720879866931755} \pi$
 & $47$ & $\num{9.903508341092198}$ & $\num{12.720879866931755} \pi$
 & $43$ & $\num{9.112052614008094}$ & $\num{10.11125529988067} \pi$
  \\
  % & \num{0.40565964089440004}
  & 0.4
  & $1826$ & $1690$
 & $101$ & $\num{10.907318296764346}$ & $\num{24.286919187476517} \pi$
 & $98$ & $\num{10.819278524062344}$ & $\num{23.258192001891626} \pi$
 & $69$ & $\num{9.89648175404852}$ & $\num{15.131487533311264} \pi$ \\ \hline
  \end{tabular}
\end{table}

\begin{figure}[H]
  \begin{subfigure}[]{0.24\linewidth}
    \includegraphics[width=\columnwidth]{data/C1_cylinder_2.0_1.0_300/0005sec_0.1ratio/cylinder_2.0_1.0_300_0th_cyclebefore.pdf}
    \caption{$z_0$, $\num{21.247393198983932} \pi$}
  \end{subfigure}
    \begin{subfigure}[]{0.24\linewidth}
      \includegraphics[width=\columnwidth]{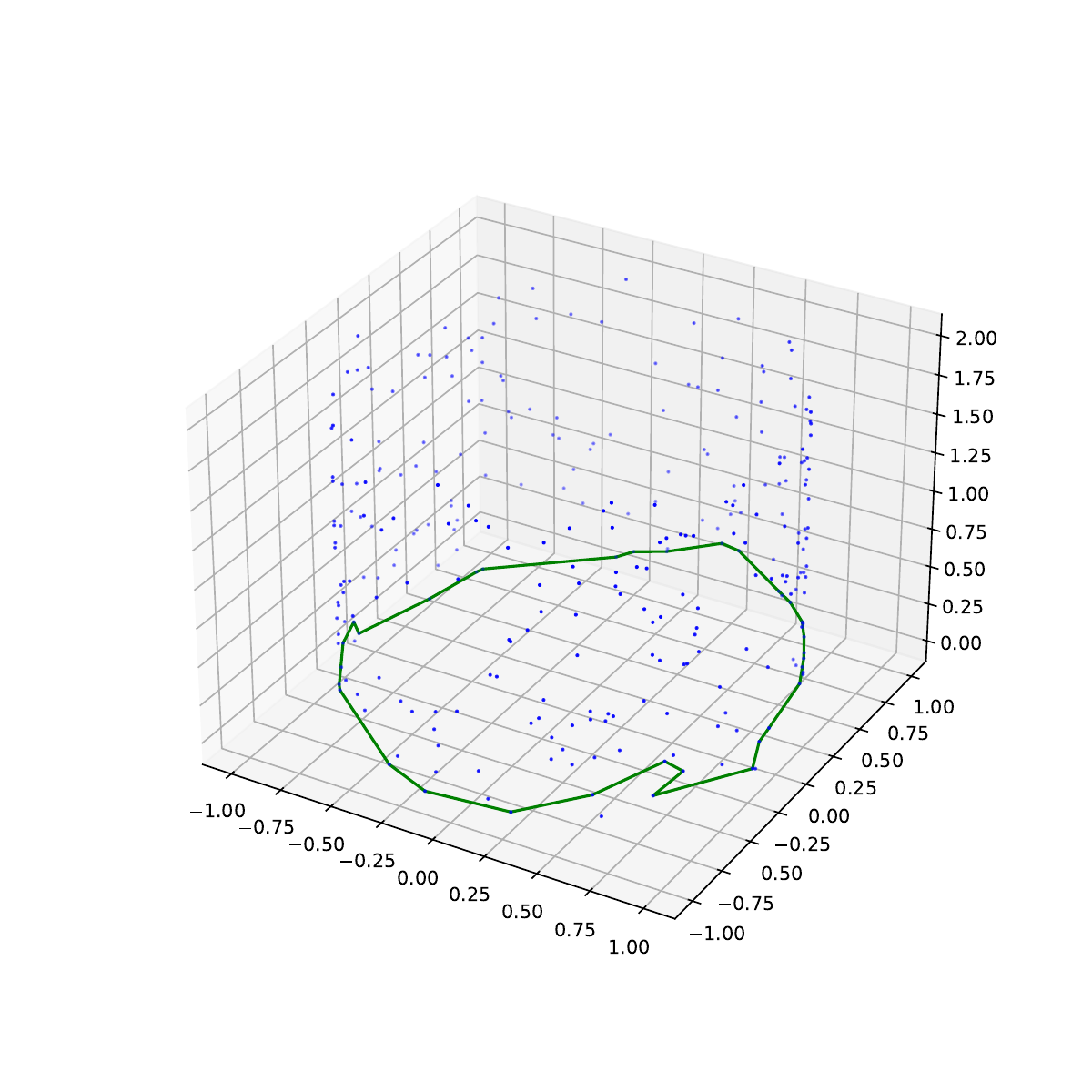}
      \caption{$z_1$, $\num{5.661044433980622} \pi$}
  \end{subfigure}
    \begin{subfigure}[]{0.24\linewidth}
      \includegraphics[width=\columnwidth]{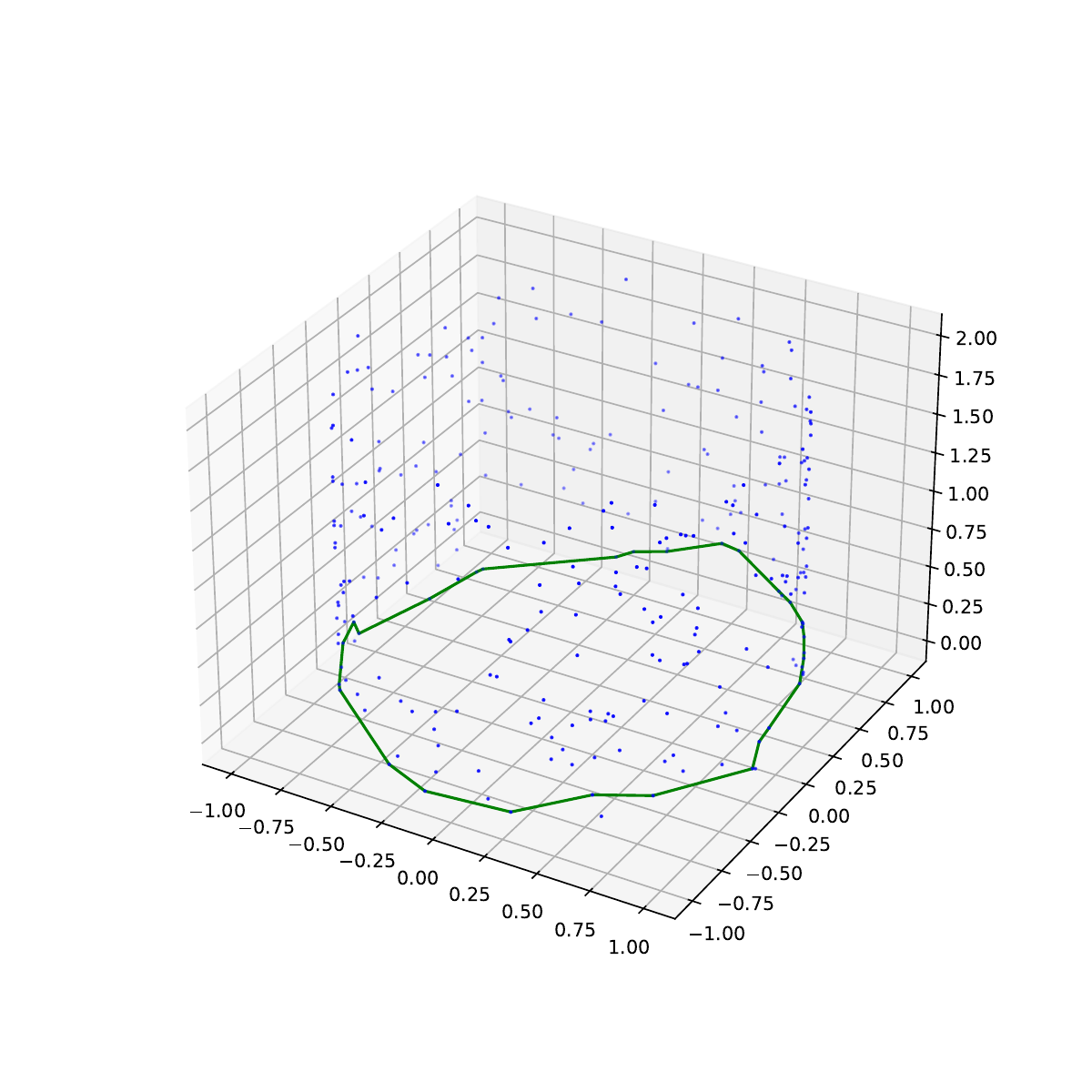}
      \caption{$z_2$, $\num{4.235763078989115} \pi$}
  \end{subfigure}
    \begin{subfigure}[]{0.24\linewidth}
      \includegraphics[width=\columnwidth]{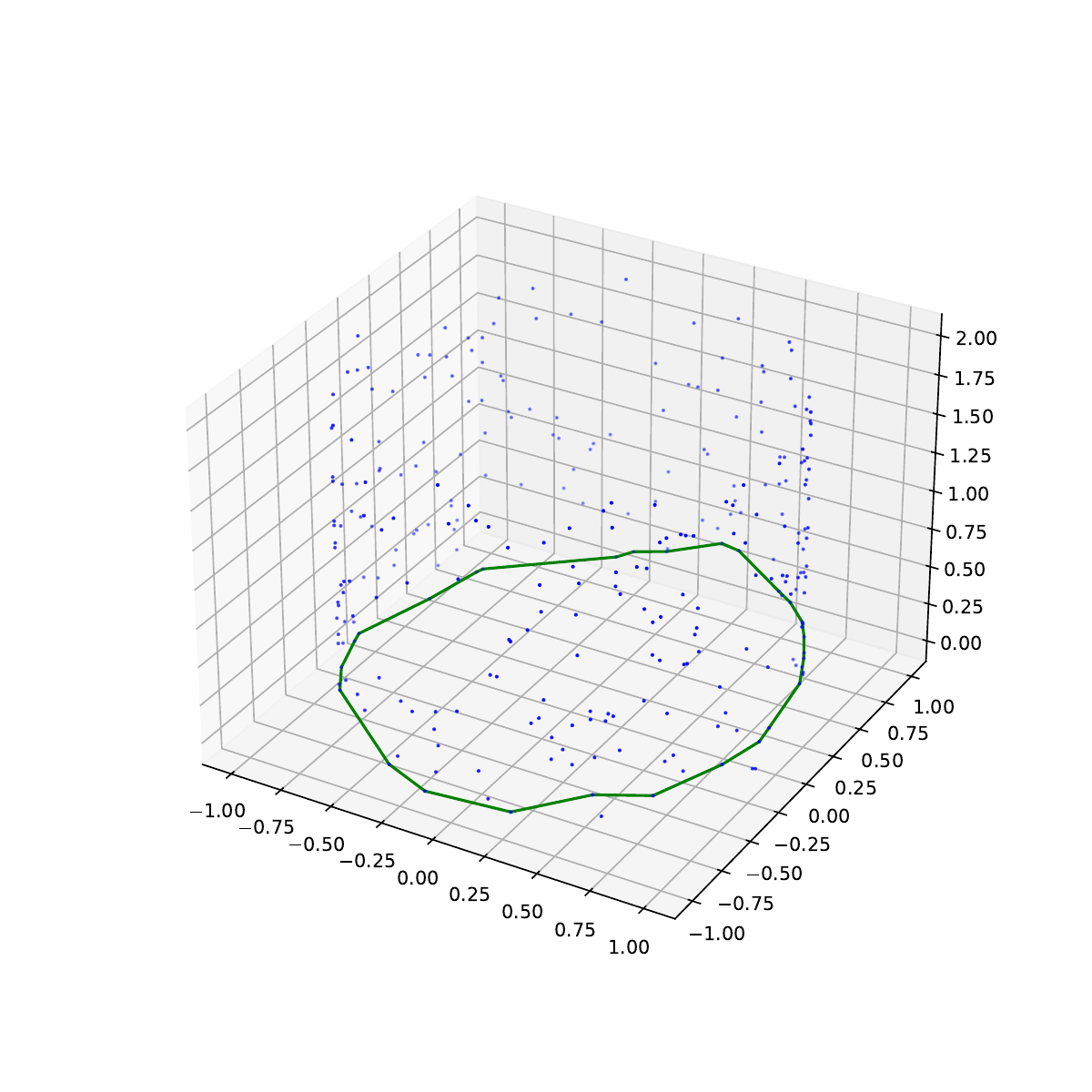}
      \caption{$z_3$, $\num{3.3733039691048523} \pi$}
  \end{subfigure}
  \caption{Initial cycle and computed solutions $z_i$,
     together with values of $\kappa(z_i)$,
    for data C1 with $t=0.1$ (first row of Table~\ref{table:cylinder}).
  }
  \label{fig:C1r1}
\end{figure}

\begin{figure}[H]
  \begin{subfigure}[]{0.24\linewidth}
    \includegraphics[width=\columnwidth]{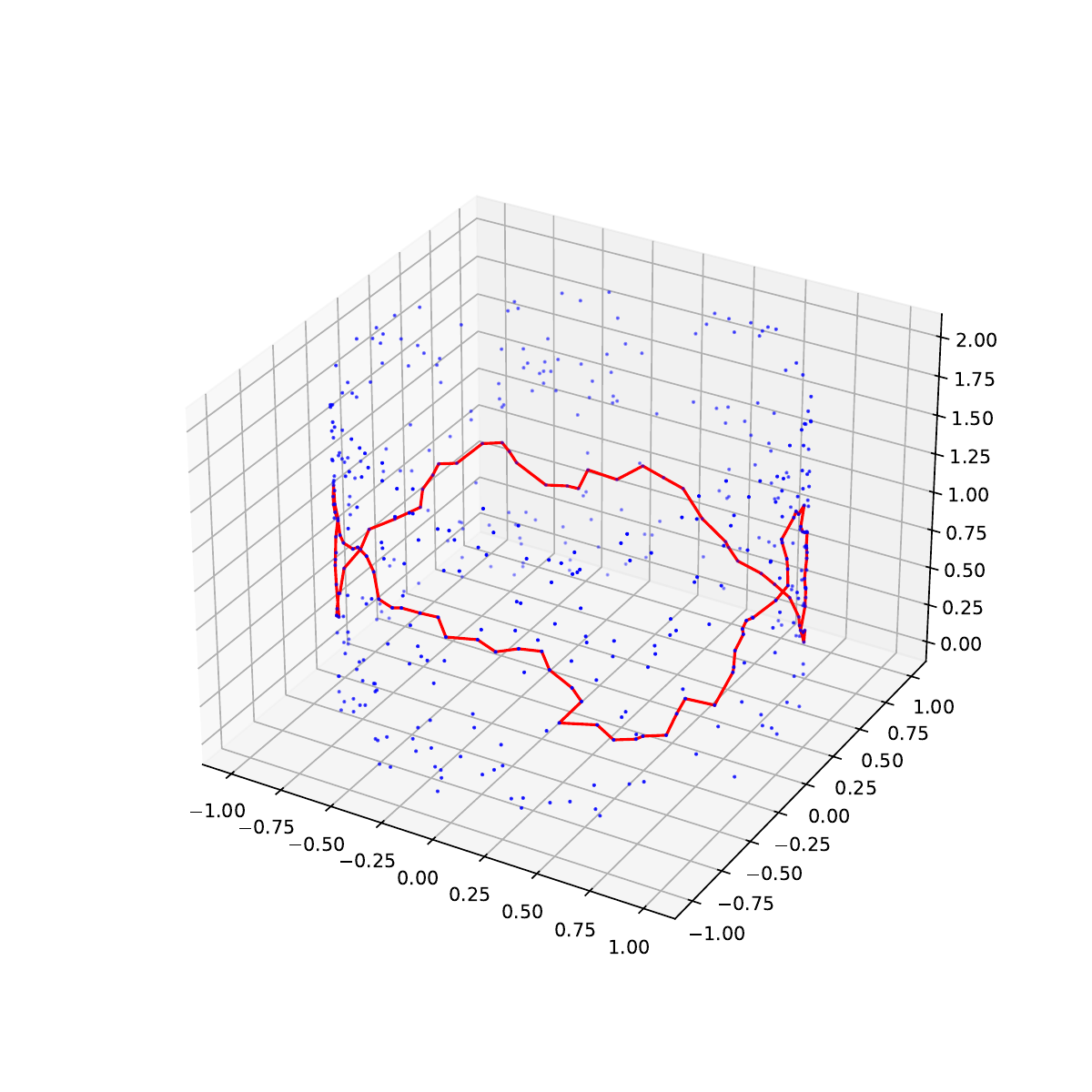}
    \caption{$z_0$, $\num{26.339357944443133} \pi$}
  \end{subfigure}
    \begin{subfigure}[]{0.24\linewidth}
      \includegraphics[width=\columnwidth]{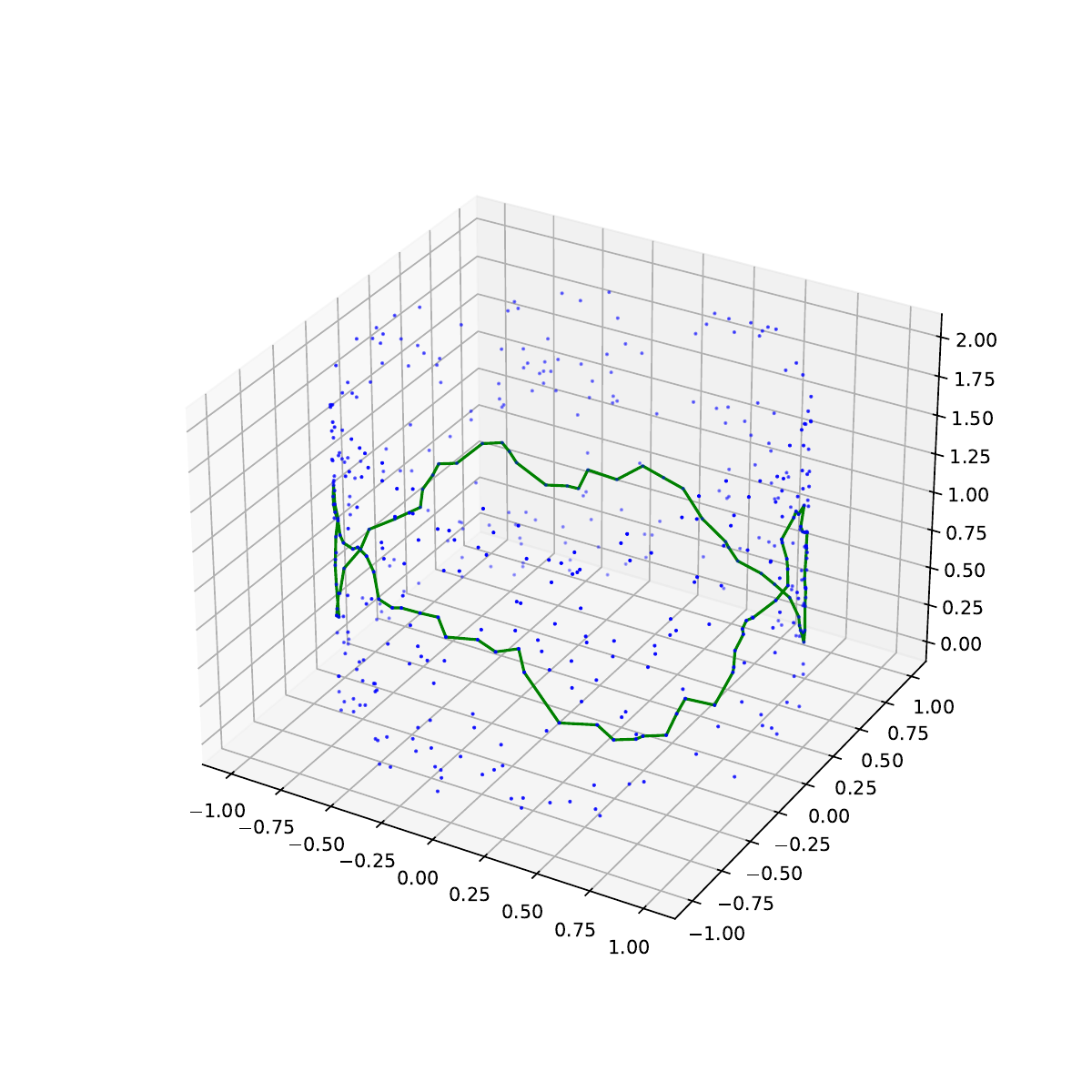}
      \caption{$z_1$, $\num{24.286919187476517} \pi$}
  \end{subfigure}
    \begin{subfigure}[]{0.24\linewidth}
      \includegraphics[width=\columnwidth]{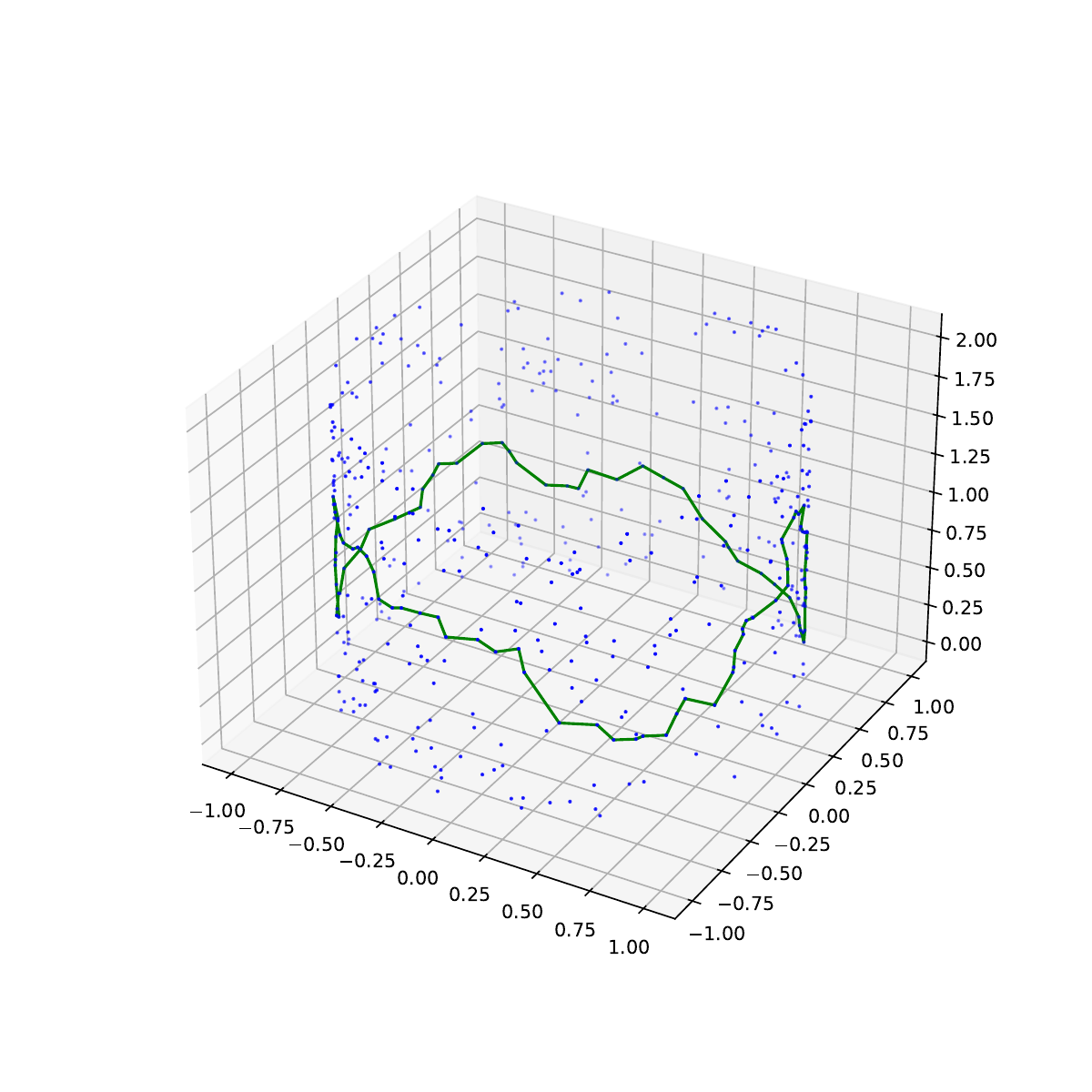}
      \caption{$z_2$, $\num{23.258192001891626} \pi$}
  \end{subfigure}
  \begin{subfigure}[]{0.24\linewidth}
    \includegraphics[width=\columnwidth]{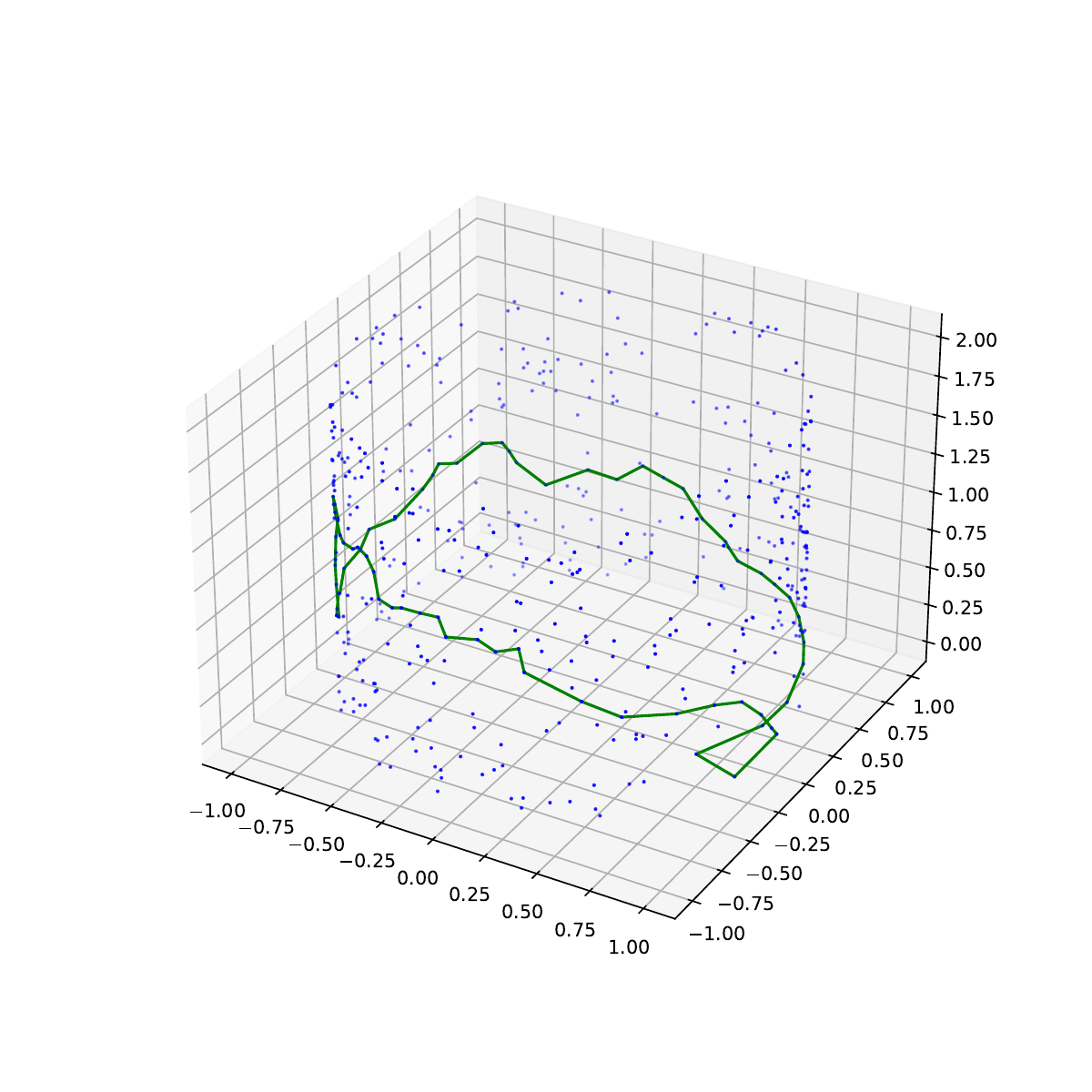}
      \caption{$z_3$, $\num{15.131487533311264} \pi$}
  \end{subfigure}
  \caption{Initial cycle and computed solutions $z_i$,
     together with values of $\kappa(z_i)$,
    for data C2 with $t=0.4$ (last row of Table~\ref{table:cylinder}).}
  \label{fig:C2r3}
\end{figure}

\FloatBarrier
% *************** END CYLINDER FLOATS ***************

To check the variation in the results, we repeat the experiment multiple times for data C2 with $t=0.2$.
In particular, we collect the final $\kappa(z_3)$ values after $15$ seconds of solve time
for {$100$} experiments. The summary statistics are as follows.
\begin{center}
  \begin{tabular}{|r|l|}
    \hline
    Worst & $\num{12.511220037839152}\pi$ \\
    Best & $\num{3.611048913392255}\pi$ \\
    Average & $\num{8.834929182210724}\pi$ \\
    Standard ddeviation & $\num{2.0590935470371594}\pi$ \\ \hline
  \end{tabular}
\end{center}
We see that there is a range of variation in the final obtained total absolute curvature (see Figure~\ref{appendix:fig:C2r2_hist_kappa} in the Appendix for a histogram).

We also perform the experiment with a longer time limit of $30$ minutes for selected parameters,
and with the larger point cloud C3 with $1000$ points.
We tabulate the results in Table~\ref{table:cylinder_more} and illustrate the solutions obtained in Figures~\ref{fig:C2_more}~and~\ref{fig:C3_more}.

\begin{table}[h]
  \caption{Additional results for the ``cylinder'' data C2 and C3.
    The solver is repeatedly paused and restarted with a time limit of 30 minutes before pausing.
    At the $i$th pause we capture the best solution found so far as $z_i$.}
  \label{table:cylinder_more}
    \begin{tabular}{|c|c|c|c| |c|c|c| |c|c|c|}
    \hline
  & & \multicolumn{2}{c||}{size of $K_r$}  & \multicolumn{3}{c||}{$i=1$ ($30$ min)} & \multicolumn{3}{c||}{$i=2$ ($60$ min)}  \\
  \hline
  % & $r$
  & $t$
 & $n_1$ & $n_2$ & $|z_1|$   & $\ell(z_1)$ & $\kappa(z_1)$ & $|z_2|$   & $\ell(z_2)$ & $\kappa(z_2)$  \\ \hline\hline
  \multirow{1}{*}{C2}
  & 0.2
  & 1677 & 1402
  & 27 & \num{6.4031038416678365} & $\num{3.302806492940285}\pi$
  & 26 & \num{6.371582468091593} & $\num{2.91173138092224}\pi$
 \\
  \hline
  \multirow{1}{*}{C3}
  & 0.2
  & 3391 & 2837
     & 51 & \num{7.208662460607396} & $\num{5.874203144730933}\pi$
     & 46 & \num{7.053745626939992} & $\num{4.946917670681228}\pi$
  \\
  \hline
  \end{tabular}
\end{table}

\begin{figure}[H]
  \centering
  \begin{subfigure}[]{0.24\linewidth}
    \includegraphics[width=\columnwidth]{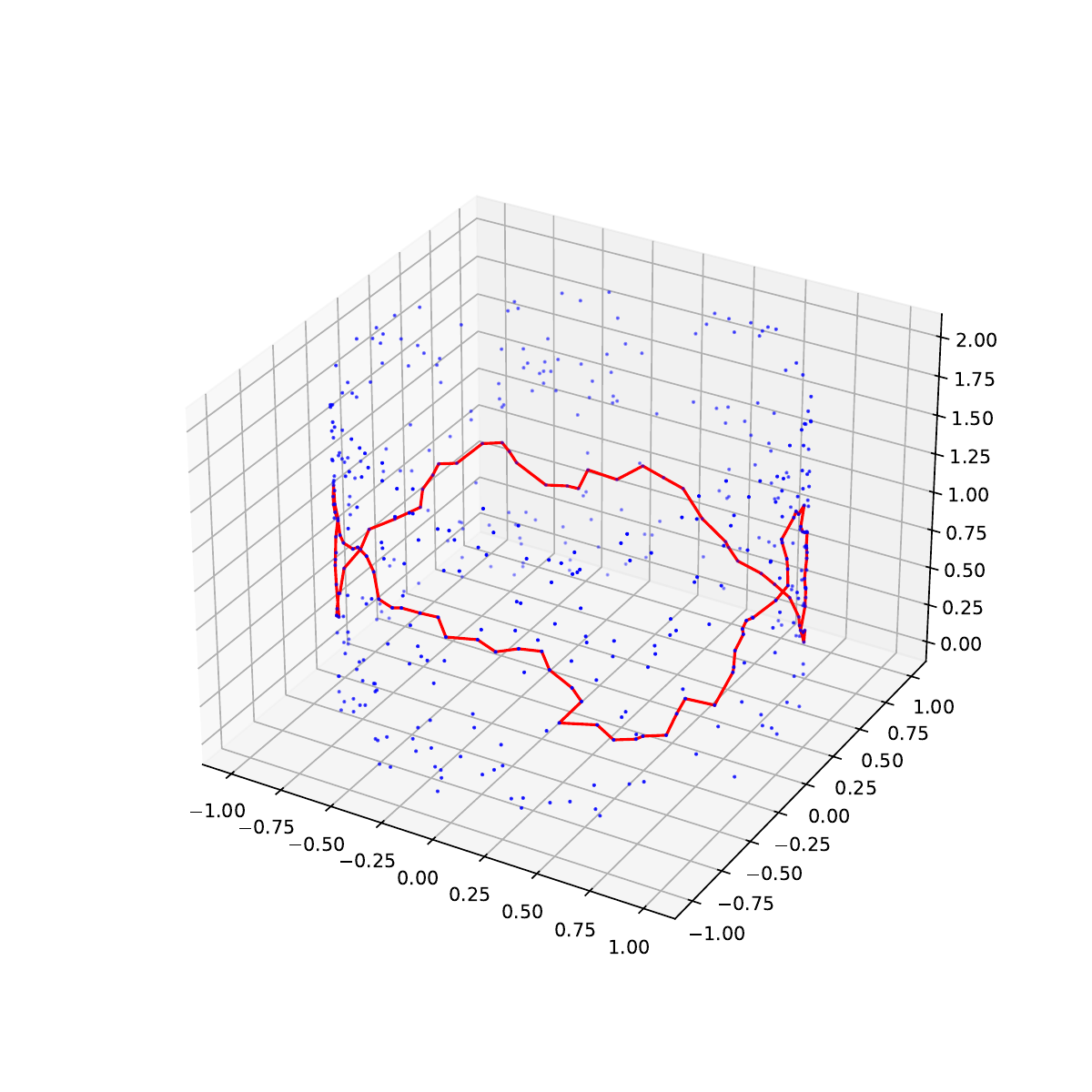}
    \caption{$z_0$\newline\hspace*{1.8em}$\kappa \approx \num{26.339357944443133} \pi$}
  \end{subfigure}
  \begin{subfigure}[]{0.24\linewidth}
    \includegraphics[width=\columnwidth]{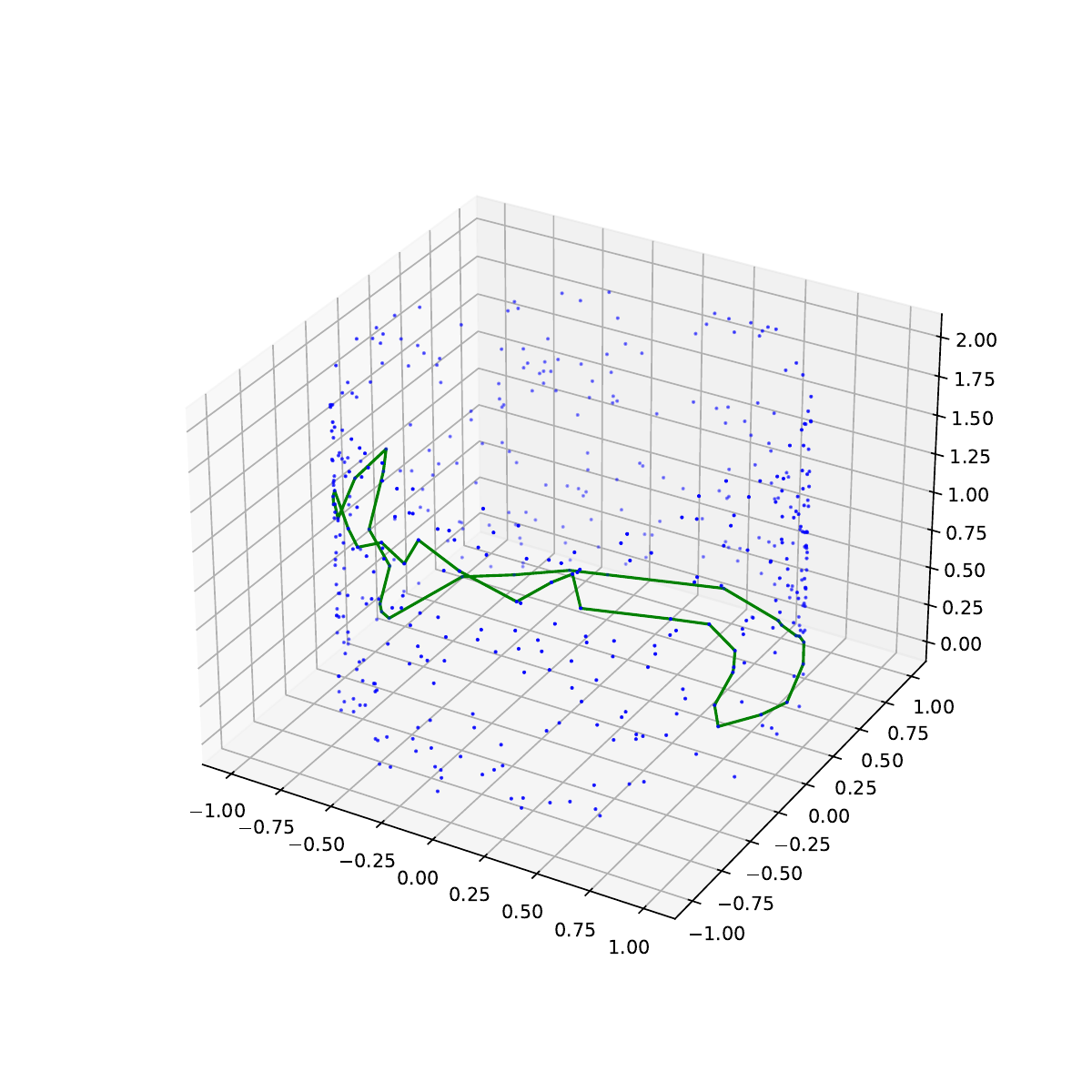}
    \caption{After $15$ sec solve.\newline\hspace*{1.8em}$\kappa \approx \num{10.11125529988067} \pi$.}
  \end{subfigure}
  \begin{subfigure}[]{0.24\linewidth}
    \includegraphics[width=\columnwidth]{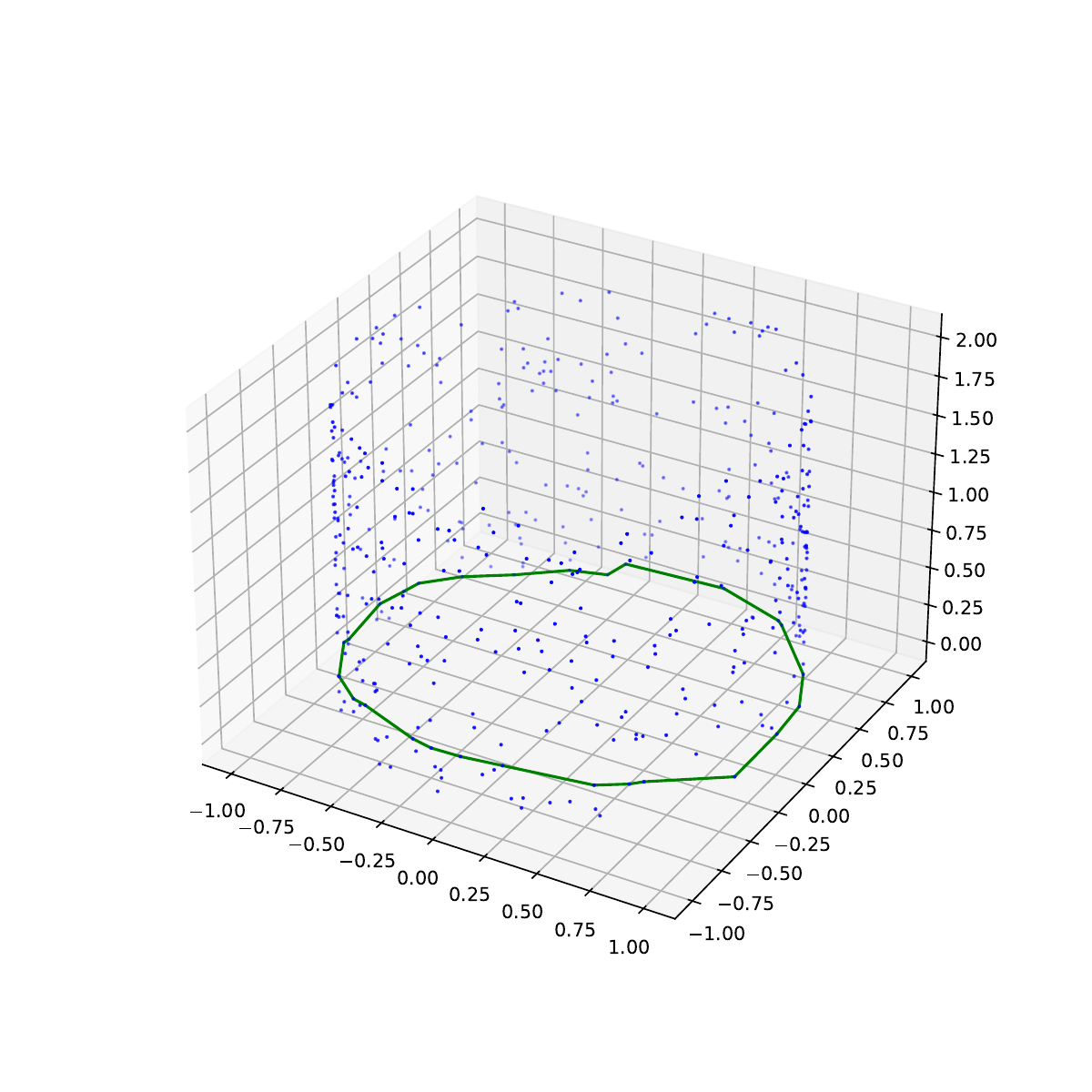}
    \caption{After $30$ min solve.\newline\hspace*{1.8em}$\kappa \approx \num{3.302806492940285}\pi$}
  \end{subfigure}
  \begin{subfigure}[]{0.24\linewidth}
    \includegraphics[width=\columnwidth]{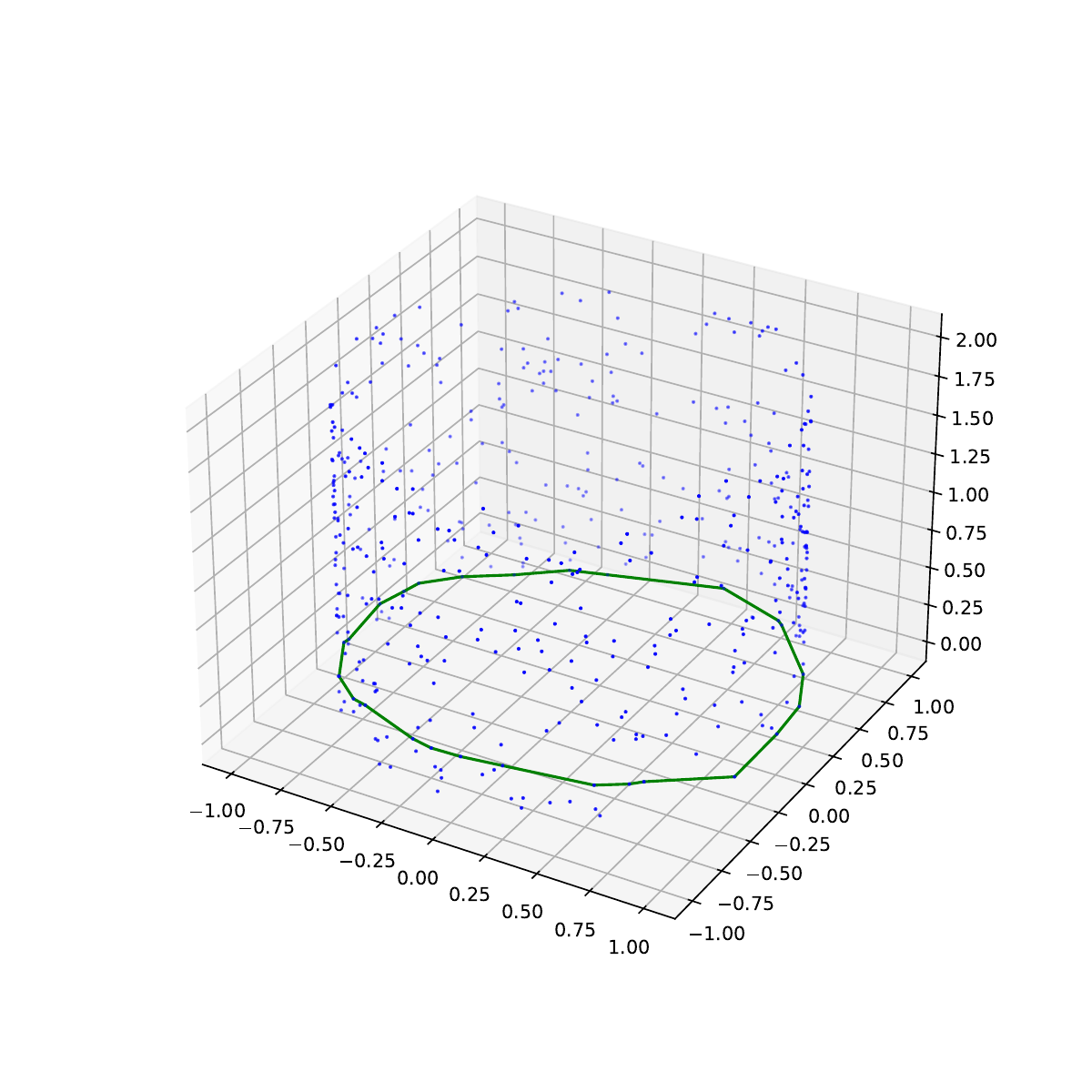}
    \caption{After $60$ min solve.\newline\hspace*{1.8em}$\kappa \approx \num{2.91173138092224}\pi$}
  \end{subfigure}
  \caption{Initial cycle and computed solutions for data C2 with $t=0.2$.}
  \label{fig:C2_more}
\end{figure}

\begin{figure}[H]
  \begin{subfigure}[]{0.24\linewidth}
    \centering
    \includegraphics[width=\columnwidth]{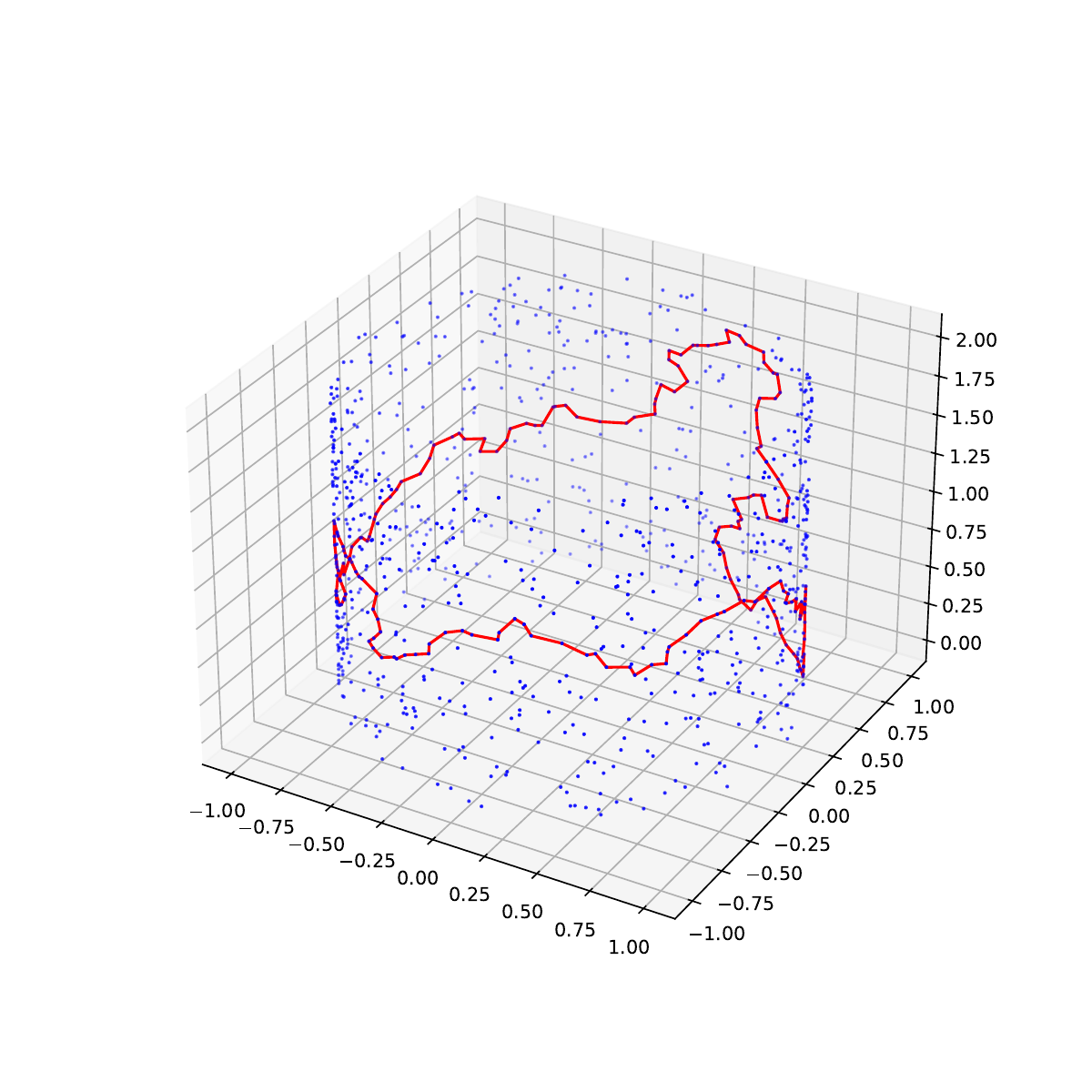}
    \caption{$z_0$\newline\hspace*{1.8em}$\kappa \approx \num{47.5274394585}\pi$}
  \end{subfigure}
  \begin{subfigure}[]{0.24\linewidth}
    \includegraphics[width=\columnwidth]{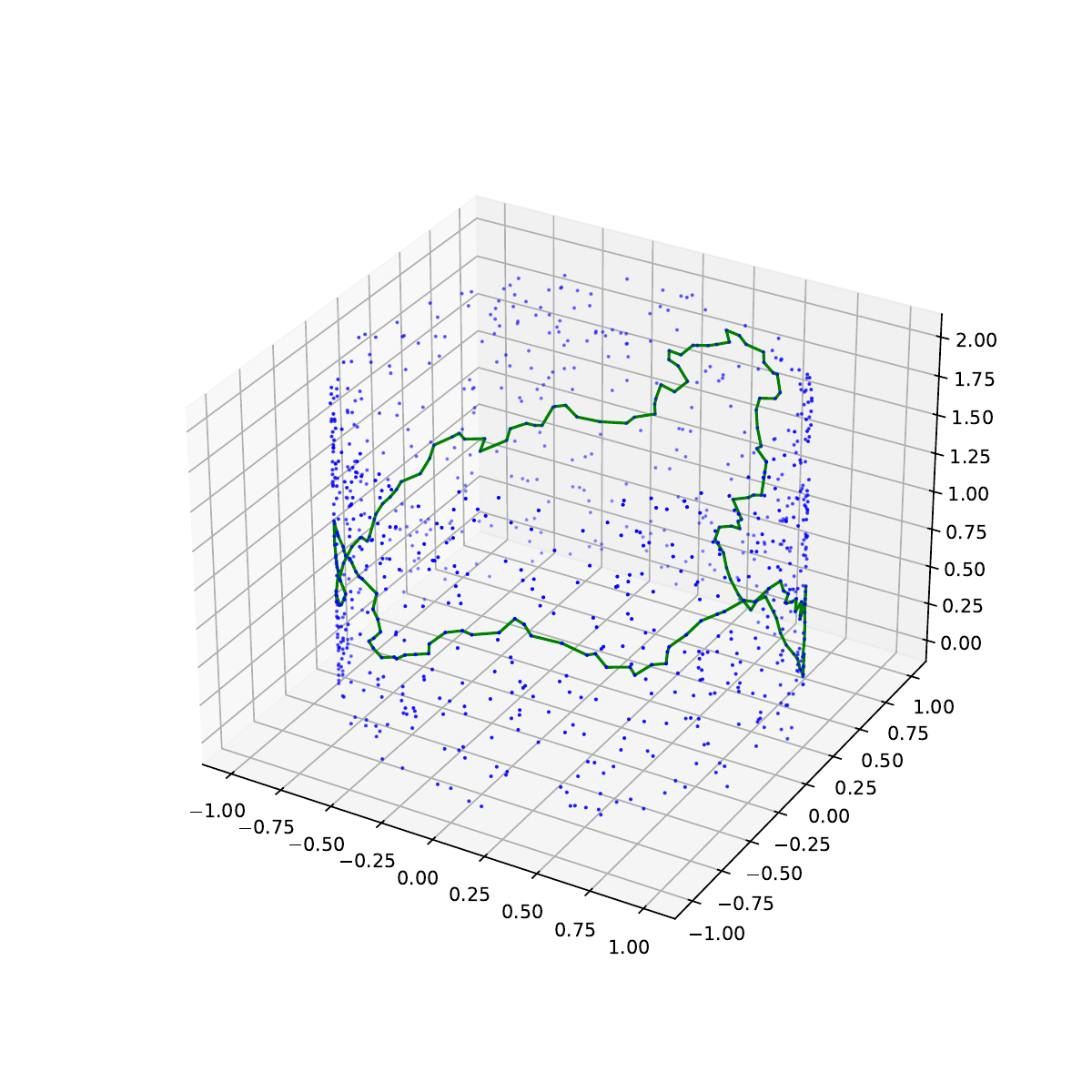}
    \caption{After $15$ sec solve.\newline\hspace*{1.8em}$\kappa \approx \num{43.66159394018956}\pi$}
  \end{subfigure}
  \begin{subfigure}[]{0.24\linewidth}
    \includegraphics[width=\columnwidth]{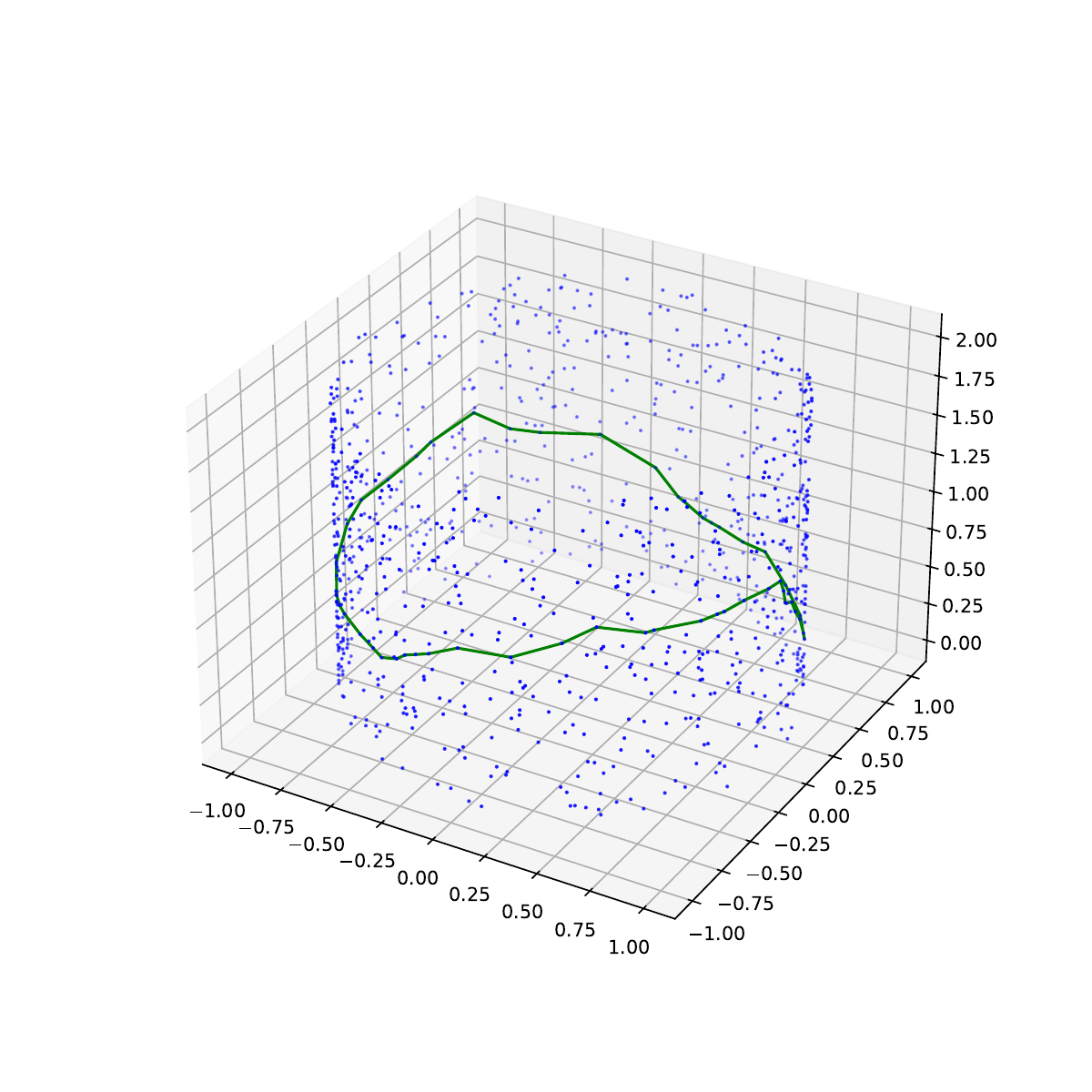}
    \caption{After $30$ min solve.\newline\hspace*{1.8em}$\kappa \approx \num{5.874203144730933}\pi$}
  \end{subfigure}
  \begin{subfigure}[]{0.24\linewidth}
    \includegraphics[width=\columnwidth]{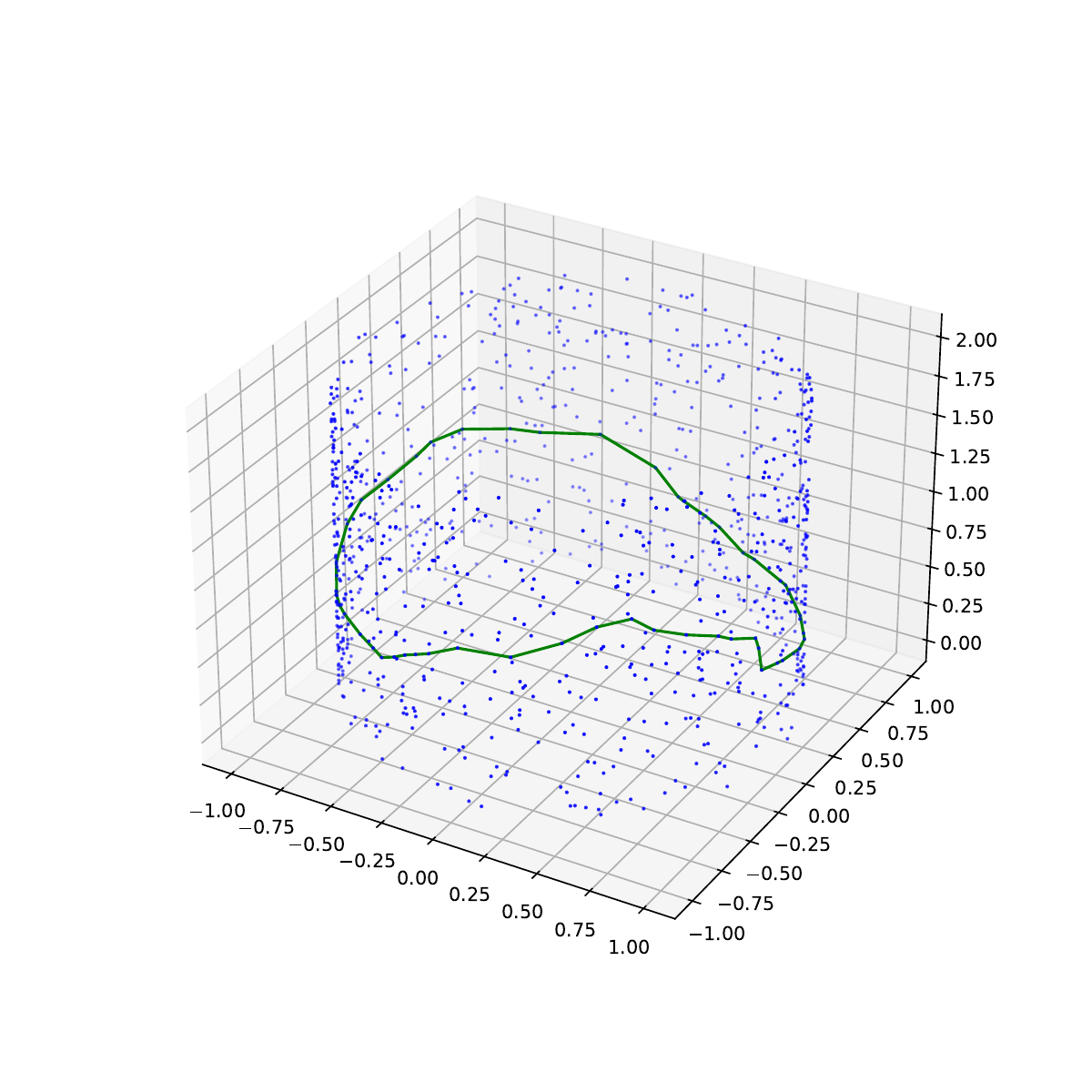}
    \caption{After $60$ min solve.\newline\hspace*{1.8em}$\kappa \approx \num{4.946917670681228}\pi$}
  \end{subfigure}

  \caption{Initial cycle and computed solutions
    for data C3 with $t = 0.2$. We see that solution (b),
    obtained after a total $15$ second time limit,
    is only slightly different from the initial cycle (a).}
  \label{fig:C3_more}
\end{figure}

\FloatBarrier

\subsection{Results on the ``slipper'' data}

% *************** SLIPPER EXPLANATION ***************

Next, in Table~\ref{table:slipper_S1}, we display the results for running the optimizer on the slipper data S1.
Again, the three subrows correspond to the choices of $t = 0.1$, $0.2$, and $0.4$ for $r := b + t(d-b)$.
For this data, we set the time limit of {1 seconds} before each pause.
For S1, we see that the solver quickly finds the cycle around the base of the slipper
with total absolute curvature  $\num{2.0570574711797383}\pi$.
We note in addition that
the solution $z_1$ has length $\ell(z_1) \approx \num{9.67741825074421}$
(for $t=0.1$ and $t=0.2$)
which is longer than the initial cycle $z_0$ with length $\ell(z_0) \approx \num{9.349799958023352}$.
In this case, decreasing total absolute curvature increased the length.
In general, \emph{the problem of minimizing length is \textbf{different} from the problem of minimizing total absolute curvature}.

Furthermore, to check the variation in output, we again do repeated experiments, for data S1 with $t=0.2$.
We obtained as solution the cycle along the base of the slipper with total absolute curvature of $\num{2.0570574711797396} \pi$ in all {$100$} runs,
suggesting that for this simple data the optimizer
is able to quickly find the global optimum consistently.

% *************** START SLIPPER S1 FLOATS ***************
\FloatBarrier

\begin{table}[h]
  \caption{Results for the ``slipper'' data S1.
    The solver is repeatedly paused and restarted with a time limit of
    $1$ second
    (the minimum possible setting for \texttt{timelimit} parameter in DOcplex).
    At the $i$th pause we capture the best solution found so far as $z_i$. We also include details about the initial cycle $z_0$ (which does not depend on $t$) for ease of comparison}
  \label{table:slipper_S1}
    \begin{tabular}{|c|c|c|c| |c|c|c| |c|c|c| |c|c|c|}
    \hline
    & & \multicolumn{2}{c||}{size of $K_r$}
  & \multicolumn{3}{c|}{$i=0$ (initial)}
  & \multicolumn{3}{c||}{$i=1$ ($1$ sec)}
  & \multicolumn{3}{c||}{$i=2$ ($2$ sec)}
      \\
      \hline
 & $t$ & $n_1$ & $n_2$
 & $|z_0|$ & $\ell(z_0)$ & $\kappa(z_0)$
 & $|z_1|$   & $\ell(z_1)$ & $\kappa(z_1)$
 & $|z_2|$   & $\ell(z_2)$ & $\kappa(z_2)$
 \\ \hline\hline
   \multirow{3}{*}{S1}
  & 0.1 & $557$ & $382$
& \multirow{3}{*}{$52$} & \multirow{3}{*}{$\num{9.349799958023352}$} & \multirow{3}{*}{$\num{4.511300712736921} \pi$}
& $39$ & $\num{9.67741825074421}$ & $\num{2.0570574711797396} \pi$
& \multicolumn{3}{c||}{(same as $z_1$)}
\\
  & 0.2 & $575$ & $416$
& & &
&  $39$ & $\num{9.67741825074421}$ & $\num{2.0570574711797396} \pi$
& \multicolumn{3}{c||}{(same as $z_1$)}
\\
  &  0.4 & $648$ & $560$
& & &
& $35$ & $\num{9.593203500503773}$ & $\num{2.0570574711797396} \pi$
& \multicolumn{3}{c||}{(same as $z_1$)}
\\ \hline
  \end{tabular}
\end{table}

\begin{figure}[H]
  \centering
  \begin{subfigure}[]{0.45\linewidth}
    \includegraphics[width=\columnwidth]{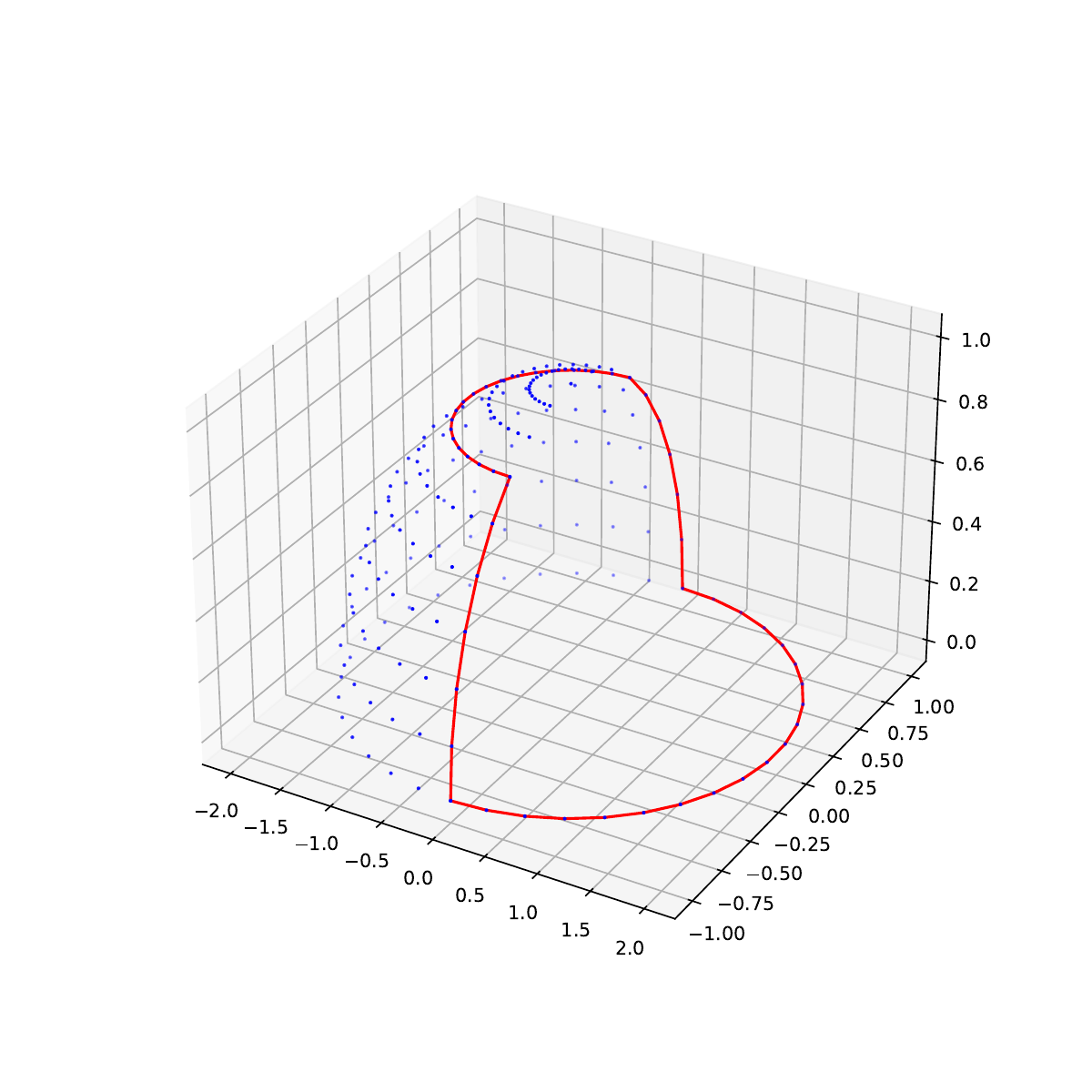}
    \caption{$z_0$, $l(z_0) \approx \num{9.349799958023352}$, $\kappa(z_0) \approx \num{4.511300712736921} \pi$}
  \end{subfigure}
    \begin{subfigure}[]{0.45\linewidth}
      \includegraphics[width=\columnwidth]{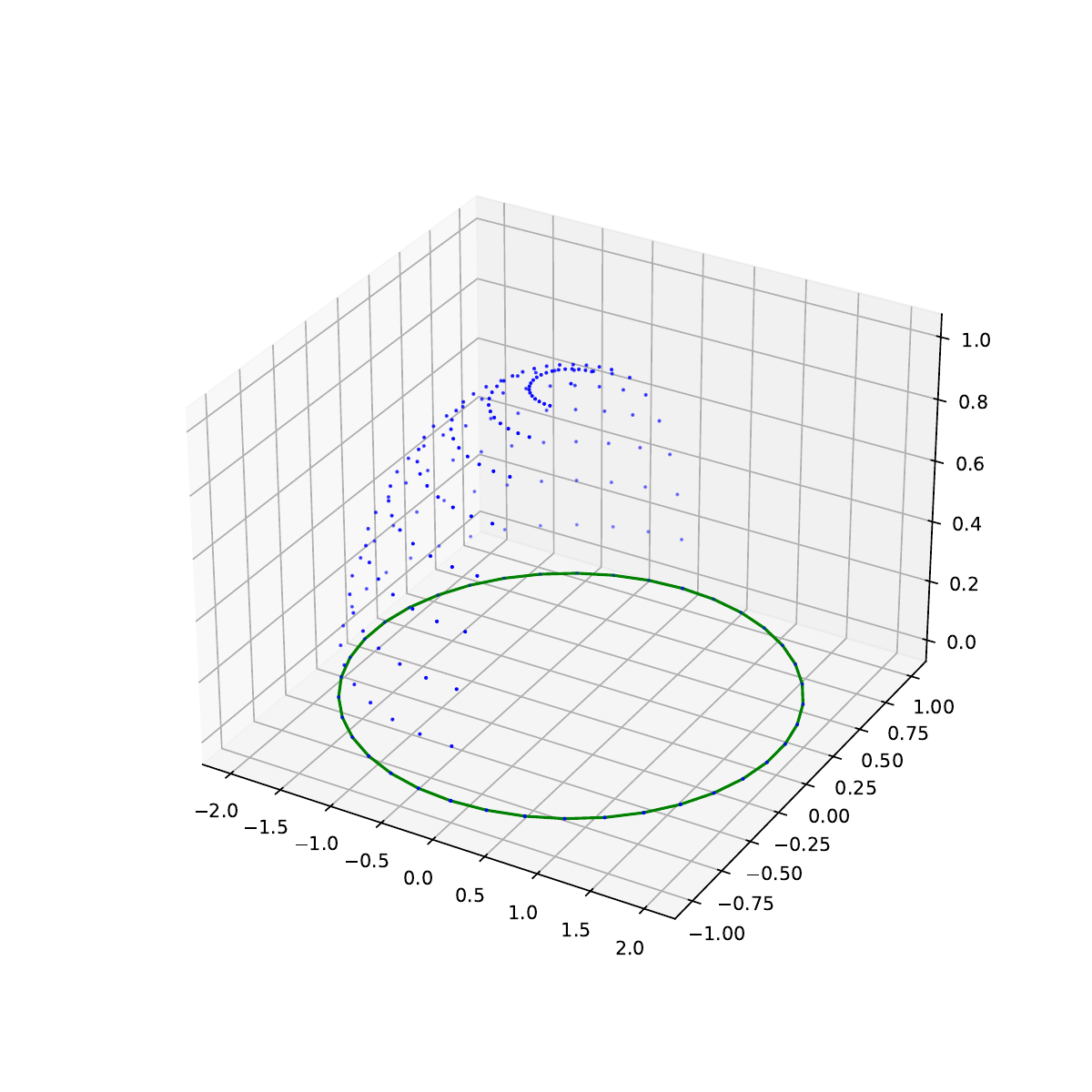}
      \caption{$z_1$, $l(z_1) \approx \num{9.67741825074421}$, $\kappa(z_1) \approx \num{2.0570574711797396} \pi$}
  \end{subfigure}
  \caption{Initial cycle $z_0$ and solution after time limit of $1$ second
    for $K_r$ with $t=0.2$ for data S1 (second row of Table~\ref{table:slipper_S1}).
    Recall that the base of the slipper is sampled from an ellipse with major axis of length $4$
    (the axis that goes across the length of the slipper goes from $-2$ to $2$) and minor axis of length $2$.
    Such an ellipse has perimeter $8\displaystyle\int_0^{\pi/2} \sqrt{1 - \frac{3}{4} \sin^2\theta} ~d\theta \approx \num{9.6884482205476761984}$ {\cite{chandrupatla2010perimeter}}.}
  \label{fig:S1r3}
\end{figure}

\FloatBarrier
% *************** END SLIPPER S1 FLOATS ***************

Next, the data S2 is the same slipper but more finely sampled. Repeating the same experiment with time limit of $1$ second did not yield acceptable results,
so we increase the time limit to $5$ seconds. The results are shown in Table~\ref{table:slipper_S2}
and in Figure~\ref{fig:S2r3}.

To again check the variation in the results, we repeat the experiment $100$ times for data S2 with $t=0.4$, and tabulate the summary statistics for the total absolute curvature $\kappa$ after $15$ seconds of solve time:
\begin{center}
  \begin{tabular}{|r|l|}
    \hline
  Worst & $\num{4.473289553809269} \pi$ \\
  Best & $\num{2.0637617408729785} \pi$ \\
  Average & $\num{3.0683382829108745} \pi$ \\
  Standard Deviation & $\num{1.0023619489316957} \pi$  \\\hline
\end{tabular}
\end{center}

Looking at the histogram (in Figure~\ref{appendix:fig:S2r3_hist_kappa} in the Appendix),
we see that out of the $100$ trials ($15$ seconds of solver time each), around a third have
a total absolute curvature close to the worst value of  $\num{4.473289553809269}\pi$,
a value close to the initial value of $\num{4.700381723548458}\pi$.

\begin{table}[h]
  \caption{Results for the ``slipper'' data S2.
    The solver is repeatedly paused and restarted with a time limit of
    $5$ seconds.
    At the $i$th pause we capture the best solution found so far as $z_i$.}
  \label{table:slipper_S2}
    \begin{tabular}{|c|c|c|c| |c|c|c| |c|c|c| |c|c|c|}
    \hline
  & & \multicolumn{2}{c||}{size of $K_r$}
& \multicolumn{3}{c|}{$i=0$ (initial)}
& \multicolumn{3}{c||}{$i=1$ (5 sec)}
& \multicolumn{3}{c||}{$i=2$ (10 sec)}
% & \multicolumn{3}{c|}{$i=3$ (15 sec)}
\\
  \hline
  & $t$ & $n_1$ & $n_2$
& $|z_0|$   & $\ell(z_0)$ & $\kappa(z_0)$
& $|z_1|$   & $\ell(z_1)$ & $\kappa(z_1)$
& $|z_2|$   & $\ell(z_2)$ & $\kappa(z_2)$
% & $|z_3|$   & $\ell(z_3)$ & $\kappa(z_3)$
\\ \hline\hline
  \multirow{3}{*}{S2}
  & 0.1
    & $1772$ & $1256$
& \multirow{3}{*}{$90$} & \multirow{3}{*}{$\num{8.893500705825458}$} & \multirow{3}{*}{$\num{4.700381723548458}\pi$}
 & $59$ & $\num{9.703560883318447}$ & $\num{2.5044743339049855} \pi$
 & \multicolumn{3}{c||}{(same as $z_1$)}
 % & $59$ & $\num{9.703560883318447}$ & $\num{2.5044743339049855} \pi$
 % & $59$ & $\num{9.703560883318447}$ & $\num{2.5044743339049855} \pi$
    \\
  & 0.2
  & $1834$ & $1376$
& & &
 & $59$ & $\num{9.7344428430453}$ & $\num{2.638238722500931} \pi$
& \multicolumn{3}{c||}{(same as $z_1$)}
 % & $59$ & $\num{9.7344428430453}$ & $\num{2.638238722500931} \pi$
 % & $59$ & $\num{9.7344428430453}$ & $\num{2.638238722500931} \pi$
  \\
  & 0.4
  & $2161$ & $2042$
& & &
 & $71$ & $\num{8.33001383769375}$ & $\num{4.407214697558679} \pi$
 & $53$ & $\num{9.640447660406496}$ & $\num{2.0637617408729803} \pi$
 % & $53$ & $\num{9.640447660406496}$ & $\num{2.0637617408729803} \pi$
  \\ \hline
  \end{tabular}
\end{table}

\begin{figure}[H]
  \centering
  \begin{subfigure}[]{0.32\linewidth}
    \includegraphics[width=\columnwidth]{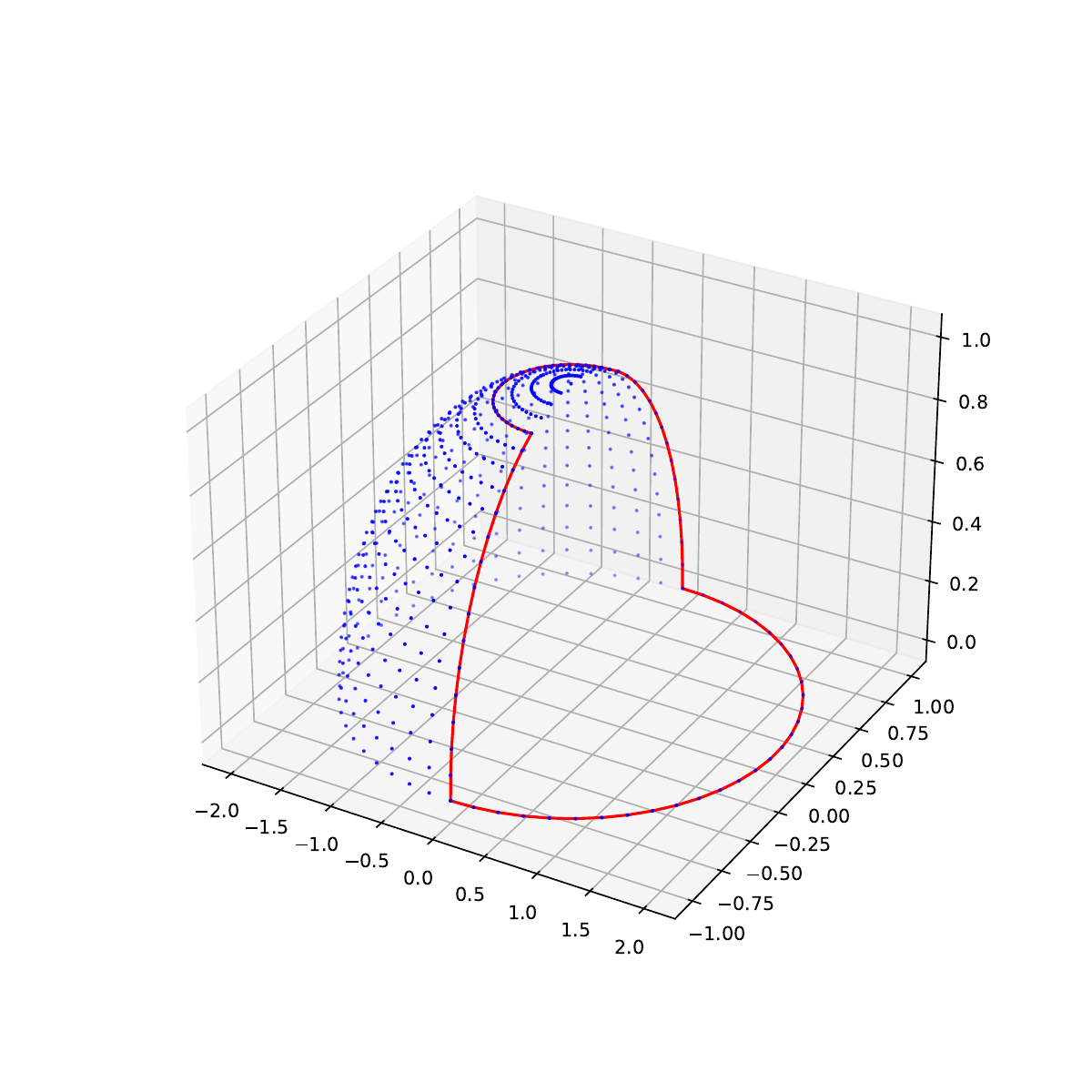}
    \caption{$z_0$, $\num{4.700381723548458} \pi$}
  \end{subfigure}
  \begin{subfigure}[]{0.32\linewidth}
    \includegraphics[width=\columnwidth]{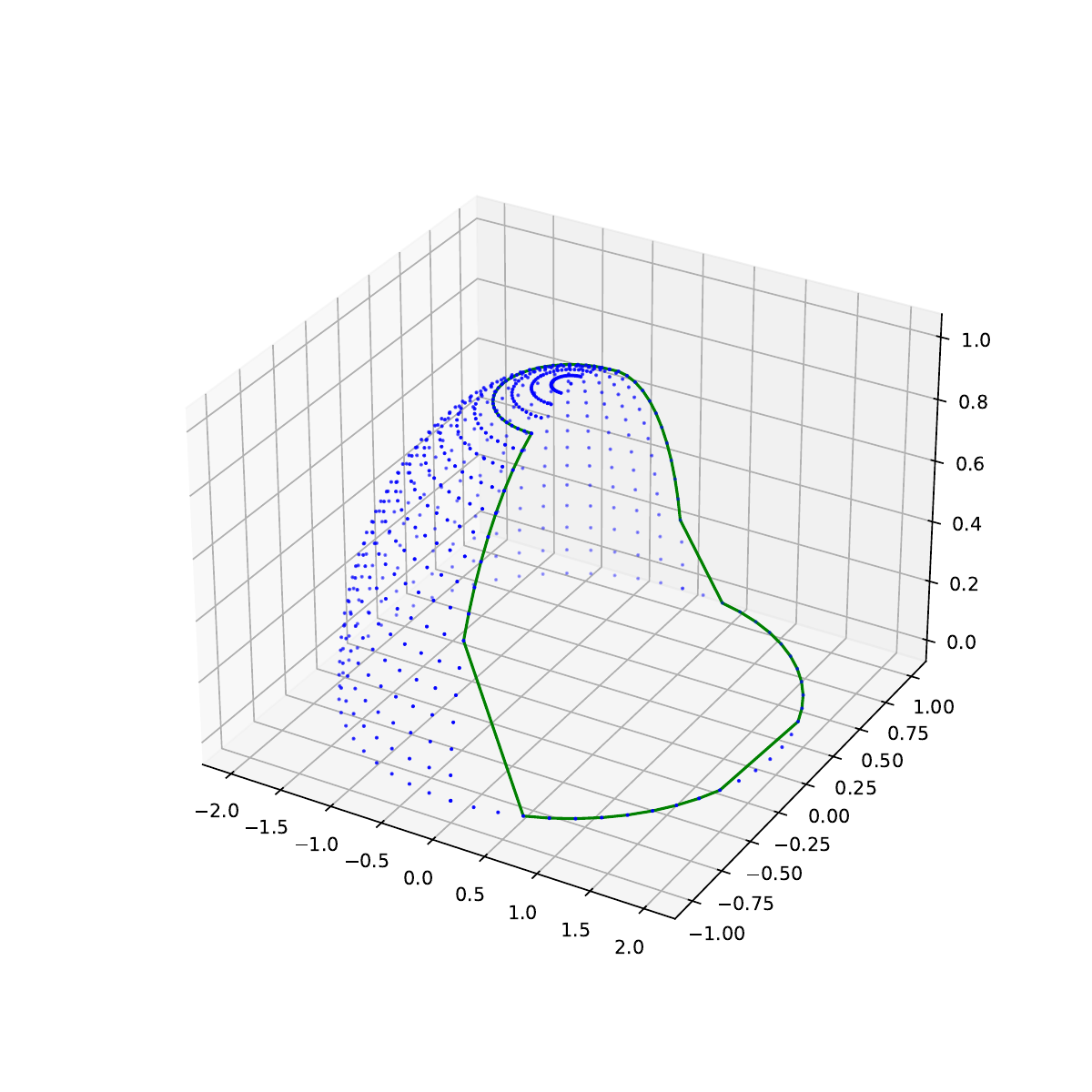}
    \caption{$z_1$, $\num{4.407214697558679} \pi$}
  \end{subfigure}
  \begin{subfigure}[]{0.32\linewidth}
    \includegraphics[width=\columnwidth]{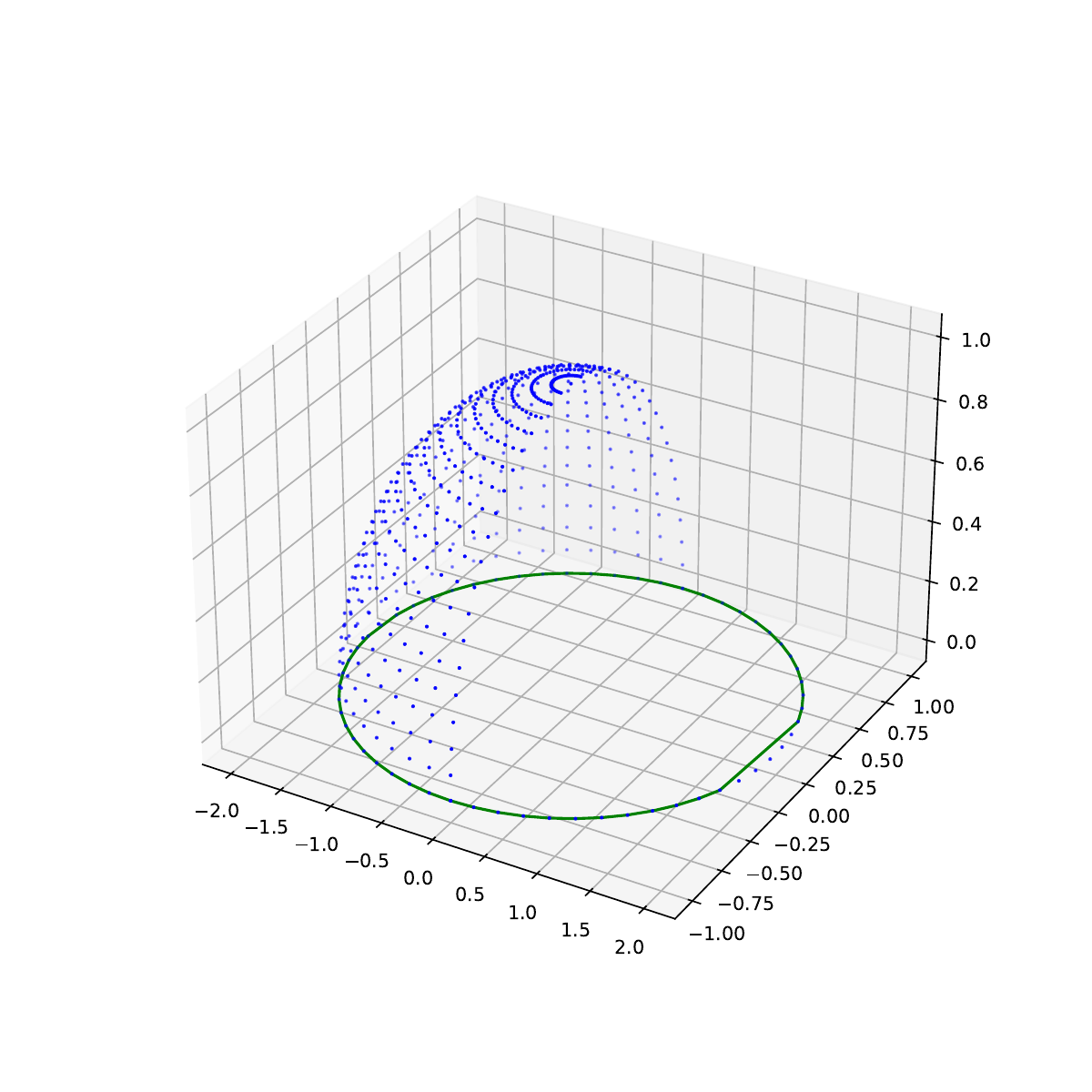}
    \caption{$z_2$, $\num{2.0637617408729803} \pi$}
  \end{subfigure}
  \caption{Initial cycle and computed solutions $z_i$, together with values of $\kappa(z_i)$,
    for data S2 with $t=0.4$ (last row of Table~\ref{table:slipper_S2}).}
  \label{fig:S2r3}
\end{figure}

\FloatBarrier

\section{Discussion}
\label{sec:discussion}

In this work, we formulate the
angle-optimal homologous cycle problem AOHCP which aims to
minimize the total absolute curvature among cycles homologous to an input $1$-cycle.
We have seen that by Theorems~\ref{thm:kappaineq}~and~\ref{thm:decompo},
the total absolute curvature
penalizes departures from planarity, convexity, and simple-ness of the cycle representative.
Through a concrete example,
we saw that the problem of minimizing total absolute curvature is different
from the well-studied problem of minimizing total length.
We also express the angle-optimal homologous problem
in a standard form of a binary quadratic optimization problem
(Problem~\eqref{opt:StandardBQP}),
and solve the problem using a commercial solver ({CPLEX} \cite{Cplex}).

Through some computational demonstrations on toy examples, we saw that for tiny point clouds,
the solver was able to find a (close-to) optimal representative very quickly.
These solutions were visually close to being planar.
Finding such planar (or close-to-planar) representatives was indeed
one of the original motivations of this work; though it turned out that
the cost function we used penalizes not just departures from planarity
but also convexity and simpleness. Is there an alternative cost function
that penalizes only departures from planarity?

As we have seen through our computations,
the solver may struggle to quickly find good solutions for
larger instances of the AOHCP, i.e.\ larger point clouds which give larger simplicial complexes.
Recall that our formulation of the problem as a binary quadratic programming problem \eqref{opt:StandardBQP} has
$2(n_1+n_2)$ binary variables and $n_1$ constraints where $n_k$ is the number of $k$-simplices.
%
% In general, over a vertex set with $n_0$ points, in the worst case a
% simplicial complex can have $n_1 = O(n_0^2)$ edges and $n_2 = O(n_0^3)$ triangles,
% leading to a binary quadratic programming problem with $2(n_1+n_2) = O(n_0^3)$ variables and $n_1=O(n_0^2)$ constraints.
Furthermore, widely-used techniques for solving binary quadratic optimization problems involve transforming it
into a mixed integer linear problem (MILP) using additional variables and constraints.
While developments in general MILP solvers have seen great improvements over its history \cite{lodi2009mixed},
our formulation may be challenging for larger $n_1$, $n_2$.

For more efficient computations, instead of using generic solvers,
a specialized algorithm that takes into account the geometric structure of the problem may be needed.
Or alternatively, a different reformulation of the angle-optimal homologous cycle problem
may be needed.

%%% Local Variables:
%%% mode: LaTeX
%%% TeX-master: "main"
%%% End:

\backmatter

% \bmhead{Supplementary information}
% If your article has accompanying supplementary file/s please state so here.

\bmhead{Acknowledgements}

E.G.E.\ is supported by
JSPS Grant-in-Aid for Scientific Research (C) (24K06846)
and
JSPS Grant-in-Aid for Transformative Research Areas (A) (22H05105).

\section*{Declarations}

\begin{itemize}
  \item Funding:
        E.G.E is supported by
        JSPS Grant-in-Aid for Scientific Research (C) (24K06846)
        and
        JSPS Grant-in-Aid for Transformative Research Areas (A) (22H05105).

  \item Competing interests:
        The authors declare no conflict of interest.

  \item Data availability: Data used the computational demonstrations is generated via the code.
  \item Code availability: \url{https://github.com/emerson-escolar/cycleflattener-v1}
  \item Author contributions:
    \textbf{E.G.E.}:
    Conceptualization,
    Methodology,
    Software,
    Formal analysis and investigation,
    Writing - original draft preparation,
    Writing - review and editing,
    Funding acquisition.
    \textbf{Y.S.}:
    Methodology,
    Software,
    Formal analysis and investigation,
    Writing - review and editing.

  \item {Note}:
    This work is based on \textbf{Y.S.}'s master's thesis,
    but with substantial improvements and more comprehensive computational experiments.
\end{itemize}

\begin{appendices}
 \section{Additional Figures}
\label{secA1}

\FloatBarrier
\begin{figure}[h]
  \includegraphics[width=\columnwidth]{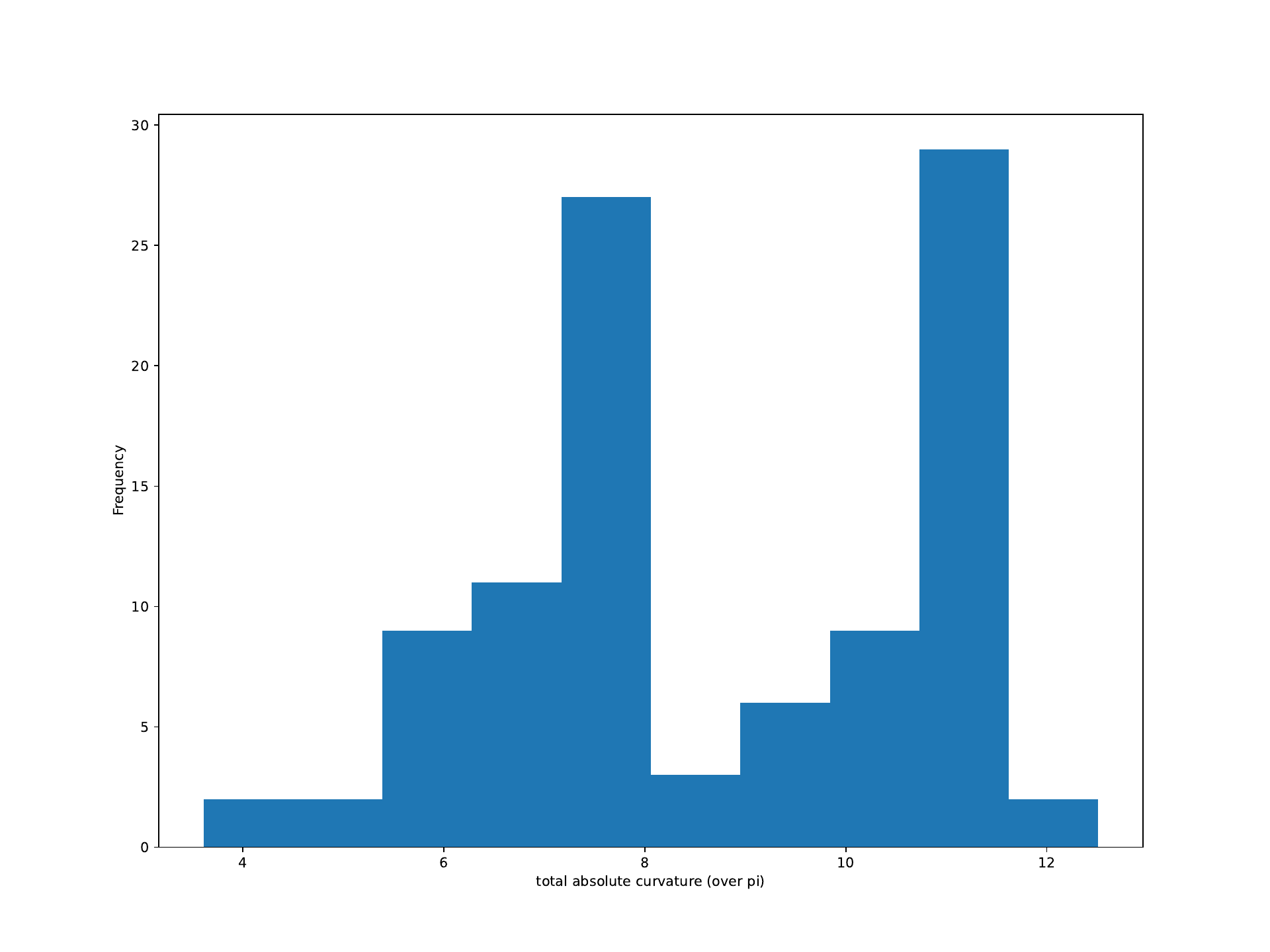}
  \caption{Histogram of total absolute curvature values (after $15$ seconds of solve time each)
    over $100$ repetitions of the experiment on data C2 with $t=0.2$.}
  \label{appendix:fig:C2r2_hist_kappa}
\end{figure}

\begin{figure}[h]
  \includegraphics[width=\columnwidth]{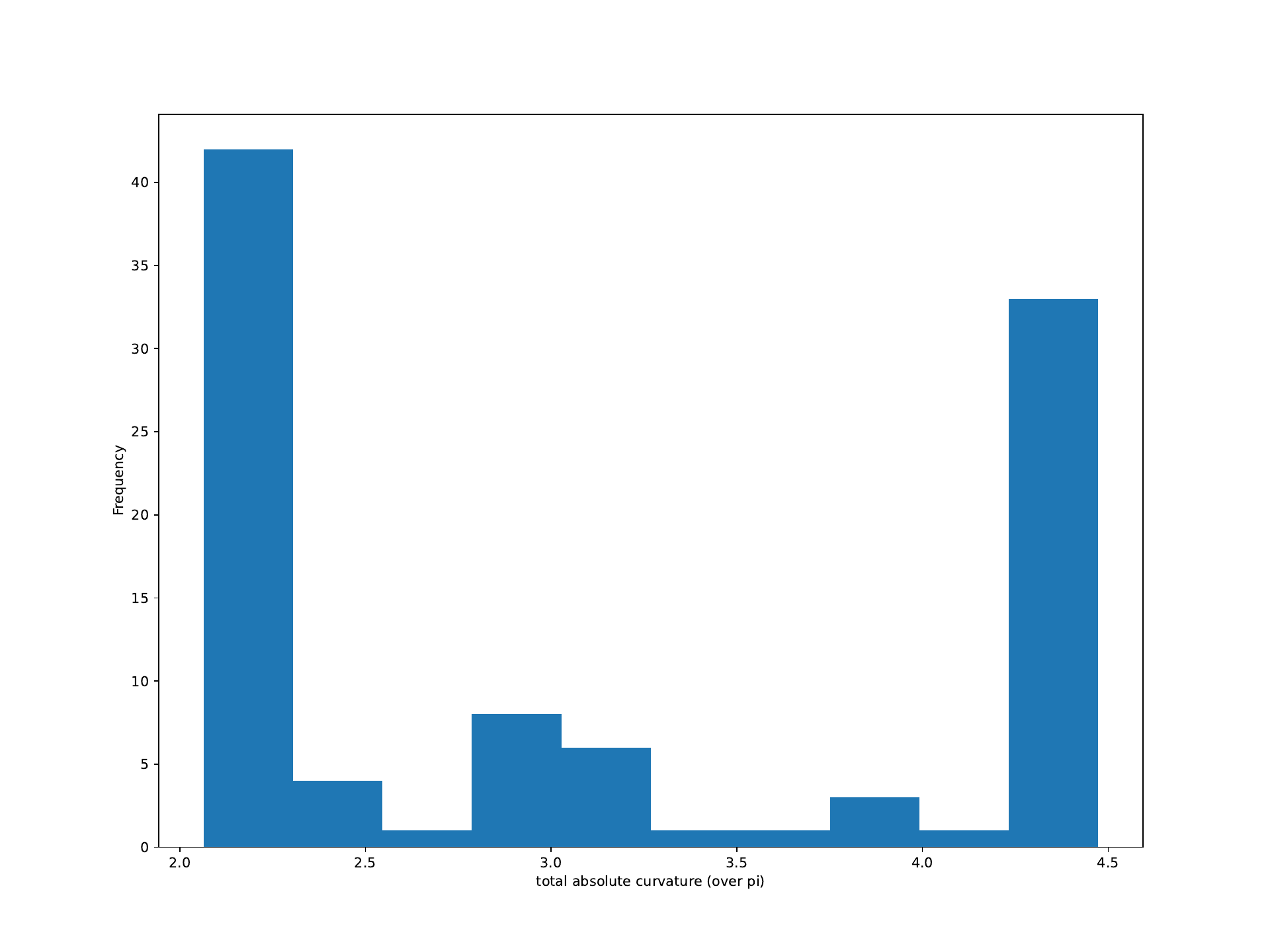}
  \caption{Histogram of total absolute curvature values (after $15$ seconds of solve time each)
    over $100$ repetitions of the experiment on data S2 with $t=0.4$.}
  \label{appendix:fig:S2r3_hist_kappa}
\end{figure}

\FloatBarrier

%%% Local Variables:
%%% mode: LaTeX
%%% TeX-master: "main"
%%% End:

\end{appendices}

\newpage

\bibliography{refs}

\end{document}